\documentclass{article}

\usepackage[utf8]{inputenc} 
\usepackage[T1]{fontenc}    
\usepackage{hyperref}       
\usepackage{url}            
\usepackage{booktabs}       
\usepackage{amsfonts}       
\usepackage{nicefrac}       
\usepackage{microtype}      

\usepackage{algorithm, algorithmic}
\usepackage{natbib}
\usepackage{graphicx}
\usepackage{multirow}
\usepackage{amsmath, bm, amsthm}
\usepackage{subcaption}
\usepackage{wrapfig}
\usepackage{xfrac}
\usepackage{sidecap}
\usepackage{adjustbox}
\usepackage{caption}

\newcommand{\deriv}[2]{\frac{\partial #1}{\partial #2}}

\newcommand{\id}{{\mathrm{id}}}

\newcommand{\var}{{\rm I\kern-.3em D}}

\newcommand{\cond}{\,|\,}

\newcommand{\dd}{\mathrm{d}}
\newcommand{\norm}[1]{\left\lVert#1\right\rVert}

\DeclareMathOperator*{\argmin}{arg\,min}

\DeclareMathOperator{\Tr}{Tr}
\newtheorem{theorem}{Theorem}
\newtheorem*{theorem*}{Theorem}
\newtheorem{lemma}{Lemma}
\newtheorem*{lemma*}{Lemma}

\newtheorem*{proposition*}{Proposition}

\newtheorem{definition}{Definition}[section]

\usepackage{xcolor}

\usepackage[accepted]{arxiv}

\icmltitlerunning{Deterministic Gibbs Sampling via ODEs}

\begin{document}

\twocolumn[
\icmltitle{Deterministic Gibbs Sampling via Ordinary Differential Equations}



\icmlsetsymbol{equal}{*}

\begin{icmlauthorlist}
\icmlauthor{Kirill Neklyudov}{uva}
\icmlauthor{Roberto Bondesan}{qualcomm}
\icmlauthor{Max Welling}{uva,qualcomm}
\end{icmlauthorlist}

\icmlaffiliation{qualcomm}{Qualcomm AI Research, Qualcomm Technologies Netherlands B.V. (Qualcomm AI Research is an initiative of Qualcomm Technologies, Inc.)}
\icmlaffiliation{uva}{University of Amsterdam}

\icmlcorrespondingauthor{Kirill Neklyudov}{k.necludov@gmail.com}

\icmlkeywords{Machine Learning, ICML}

\vskip 0.3in
]



\printAffiliationsAndNotice{}  

\begin{abstract}
Deterministic dynamics is an essential part of many MCMC algorithms, e.g. Hybrid Monte Carlo or samplers utilizing normalizing flows. This paper presents a general construction of deterministic measure-preserving dynamics using autonomous ODEs and tools from differential geometry.  We show how Hybrid Monte Carlo and other deterministic samplers follow as special cases of our theory. We then demonstrate the utility of our approach by constructing a continuous non-sequential version of Gibbs sampling in terms of an ODE flow and extending it to discrete state spaces. We find that our deterministic samplers are more sample efficient than stochastic counterparts, even if the latter generate independent samples. 
\end{abstract}

\section{Introduction}
\label{sec:intro}

Markov Chain Monte Carlo (MCMC) \citep{metropolis1953equation} has been the workhorse for sampling from probability distributions for which exact inference is intractable. In particular, in the Bayesian literature MCMC is used to sample parameter values from the posterior distribution $p(\theta|\mathcal{D})$ of the model parameters $\theta$ given data $\mathcal{D}$. The usual design of a Markov chain involves a stochastic transition kernel $\tau(x'|x)$ that has the target density as it's equilibrium distribution. One simple way to satisfy that requirement is to let the kernel satisfy detailed balance which can be easily achieved by including a Metropolis-Hastings accept-reject step in the kernel. 

Kernels that satisfy detailed balance are reversible and known to mix slowly. It is much preferred if we can add irreversible components to the sampler, see e.g. \citep{turitsyn2011irreversible}. When designing irreversible kernels one can in fact take the limit where the kernel becomes completely deterministic. While ergodicity of such samplers is not easy to prove, they can in some cases be designed to sample very efficiently from the target distribution, converging as fast as $1/T$ rather than the usual $1/\sqrt{T}$, where $T$ is the number of iterations of the sampler. 

Examples of samplers with deterministic components exist in the literature, most notably Hybrid Monte Carlo \citep{duane1987hybrid} (which still includes a MH accept-reject to correct for integration errors and resamples the momenta to achieve ergodicity). 

In this paper we make two contributions. Firstly, we develop the general theory of designing deterministic samplers from ordinary differential equations. We find that that the class of such samplers is in fact very large and can be characterized by the set of divergence free vector fields: every divergence free vector field corresponds to one such sampler (but may not be ergodic). The second contribution is to apply this to Gibbs sampling and extend that to the discrete domain. The resultant deterministic samplers are highly efficient, even when we compare to samplers that sample independently from the target distribution. 
In Fig. \ref{fig:portrait}, we illustrate this sampler on the portrait of Josiah Willard Gibbs.

\begin{figure}[t]
    \centering
    \includegraphics[width=0.48\textwidth]{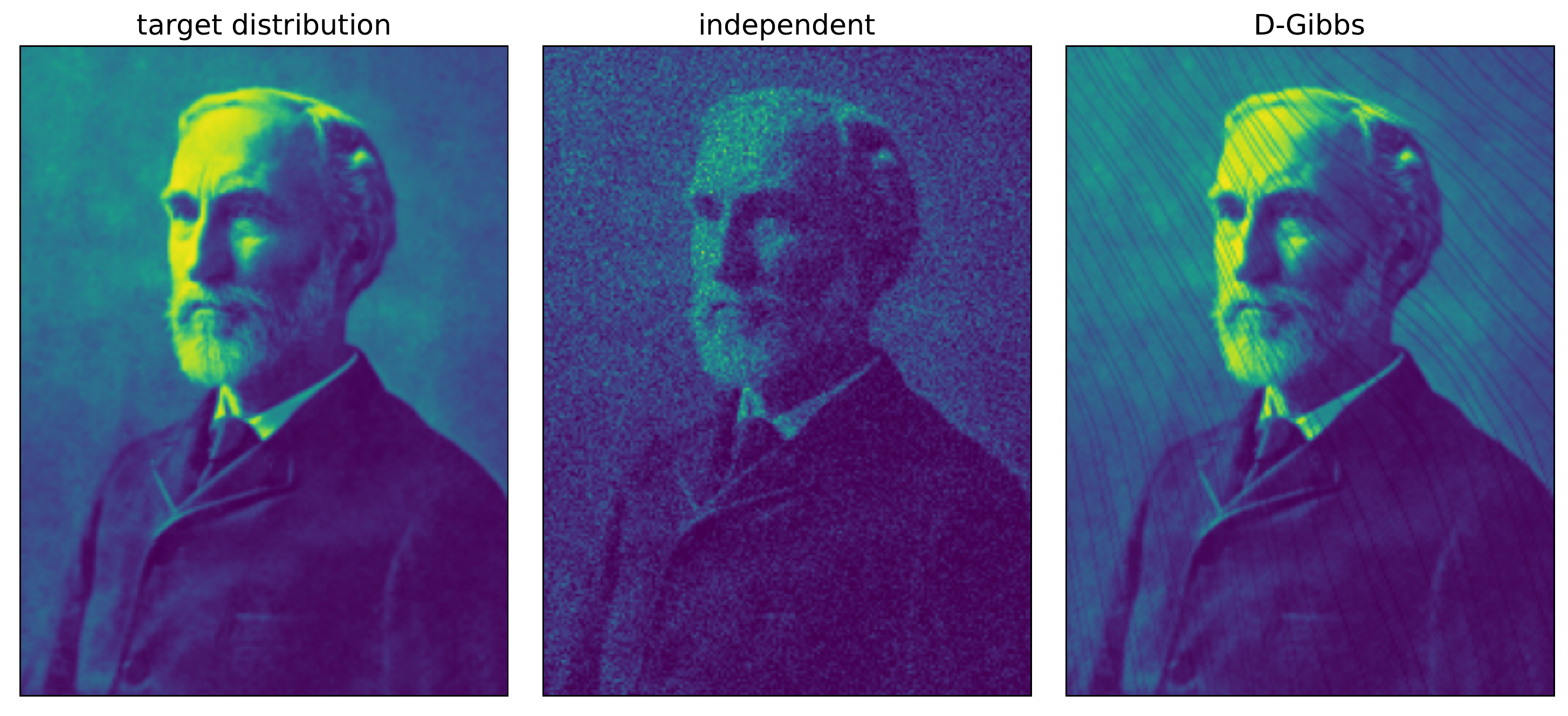}
    \caption{Sampling Gibbs with Gibbs: Histogram of samples from the discrete distribution (leftmost picture). Each pixel value of the grayscale image ($237\times 178$) is treated as an unnormalized probability. The middle histogram is generated by sampling independently from a multinomial distribution. The rightmost histogram depicts samples from the proposed deterministic Gibbs algorithm.}
    \label{fig:portrait}
    \vspace{-10pt}
\end{figure}

\vspace{-10pt}
\section{Deterministic Sampling using ODEs}
\label{sec:approach}

In this section we will derive a ``deterministic sampling'' procedure, by integrating a simple (autonomous) ODE equation 
\begin{align}
    \frac{\dd x}{\dd t} = v(x).
    \label{eqn:ODE}
\end{align}
This describes the flow of ``particles'' at every position $x$ and we demand that the flow leaves the target distribution invariant. We will show that the change in the probability density is described by the Liouville equation, which results in a simple condition for invariance.

Before we start the derivation we mention that deterministic samplers are a special class of MCMC samplers which usually design a Markov chain through \emph{stochastic} transition functions $\tau$: $p'(x') = \int dx ~\tau(x'|x) p(x)$, where $\tau$ is a conditional probability distribution. We then require invariance of the density $p'=p$ after the transition, which put constraints on $\tau$. 

One convenient class of transition kernels are the ones that satisfy detailed balance: $\tau(x'|x)p(x) = \tau(x|x')p(x')$. Since $\tau$ is its own inverse we say that these transition kernels are reversible. However, we can design (stochastic or deterministic) kernels that do not satisfy detailed balance but do leave the target distribution invariant. These are often called irreversible transition kernels (which we think is a bad name because we can in fact run the process backward for which we need the time-reversed kernel $\Tilde{\tau}(x|x') = \tau(x'|x)p(x)/p(x')$). When we consider deterministic samplers we are always considering ``irreversible'' kernels that map one point to one other point using the ODE \eqref{eqn:ODE}. The reverse kernel is implemented by simply running the ODE backward in time. 

We note that we require ergodicity from our samplers to make sure they visit every state of our configuration space with some finite probability. This is easy to achieve with stochastic samplers in theory by requiring that that the transition kernel $\tau$ has non-vanishing probability mass everywhere. However, this question is more complex for deterministic samplers as we now need to require that the ODE visits every part of state space and will not end up in some periodic orbit. We will say more on this later in the paper. 

Ergodic irreversible samplers are particularly attractive because they tend to mix much faster than reversible samplers \cite{turitsyn2011irreversible}. The difference can be understood with an analogy: reversible kernels mix through a process similar to diffusion, while irreversible samplers mix through a process similar to convection, which is much more efficient. Deterministic samplers have the added benefit of avoiding inefficiencies due to randomness.  

We now derive the condition for invariance of a probability measure under the ODE flow.

\begin{theorem}[Liouville's theorem]\label{thm:liouville}
A density $p(x)$ is invariant under the ODE flow $\dot{x} = v(x)$ if
\begin{align}
    {\rm div}(pv) = 0 \,.
\end{align}
\end{theorem}
\begin{proof}
The flow $g$ generated by the vector field $v$ from an initial value $x_0$ at time $0$ to time $t$ is
\begin{align}
    x(t) = g(x_0,t) = x_0 + \int_{0}^t v(x(t')) \dd t'
    \,.
\end{align}
Under the map $x_0 \mapsto g(x_0,t)$ the density $p(x)$ transforms with the well known ``change of variables'' formula:
\begin{align}
    & p(x_0) = p(g(x_0,t)) \bigg|\det\left(\deriv{g(x_0,t)}{x_0}\right)\bigg|,
    \label{eq_app:measure-preservings}
\end{align}

To enforce invariance of the density $p(x)$ we demand that at first order in $\epsilon$, $p(g(x_0,\epsilon)) = p(x_0)$. 
We linearize the flow
\begin{align}
    g(x_0,\epsilon) &= g(x_0,0) + \deriv{g(x_0,t)}{t}\bigg|_{t=0}\epsilon + o(\epsilon) \\
    &= x_0 + v(x_0)\epsilon + o(\epsilon),
\end{align}
so that, using $\det(\exp(A)) = \exp(\Tr(A))$,
the Jacobian is
\begin{align}
\det\left(\deriv{g(x_0,t)}{x_0}\right)
&=
\det\left(
\id
+
\deriv{v(x_0)}{x_0}
\epsilon
+ o(\epsilon)
\right)
\\
&=
1 + \epsilon\, \text{div}(v(x_0)) + o(\epsilon)
\,.
\end{align}
Plugging this into \eqref{eq_app:measure-preservings}, together with $p(g(x_0,\epsilon)) = p(x_0)+\epsilon \frac{\dd p}{\dd t}$, we get to first order in $\epsilon$:
\begin{align}
    \frac{\dd p}{\dd t} + p(x_0)\text{div}v(x_0)
    =
    0
    \label{eqn:Liouville},
\end{align} 
or, rewriting the time derivative in terms of $v$,
\begin{align}
  \nabla p\frac{\dd x}{\dd t} + p(x)\text{div}v(x) = \text{div}(pv) = 0 
  \label{eq:contin_stationary}.
\end{align}
\vspace{-10pt}
\end{proof}
This ``conservation of probability'' is known as the Liouville equation.
For a formal treatment, see e.g. \citep{fomin1981ergodic}.

It is perhaps surprising that, apart from ergodicity, the only thing we need to require for invariance is the simple condition $\text{div}(pv) = 0$. Thus, for any divergence free vector-field $w$ we can simply set $v=w/p$ and flow according to the ODE of eqn. \ref{eqn:ODE} and be guaranteed that the target density is invariant. Intuitively, this means that for any cell in state space the amount of probability flowing in and out of that cell is the same. Different velocity fields $v$ result in different mixing properties of the resultant sampler but the design space of all such fields is clearly very large. In section \ref{sec:hodge} we characterize the class of all divergence free vector fields using notions from differential geometry.    


\subsection{An Example: Hamiltonian Monte Carlo}
\label{sec:hmc}

We will next describe a simple and celebrated example of the above method.

Consider the joint distribution $p(s) = p(x,y) = p(x)p(y)$, where $p(x)$ is the target density, $p(y)$ is an auxiliary distribution, and $s = (x,y)$ is a tuple of two vectors.
Then the following dynamics preserves $p(s)$.
\begin{align}
    \frac{\dd s}{\dd t} = v(s) = \frac{1}{p(s)} \begin{bmatrix}\nabla_y H(s)\\-\nabla_x H(s)\end{bmatrix},
\end{align}
where $H(s)$ is an arbitrary differentiable function, called the Hamiltonian from the role it plays in classical mechanics.
The divergence-free property of Hamiltonian vector fields can be easily checked:
\begin{align}
    &\text{div} (v(s)p(s)) = \text{div} \begin{bmatrix}\nabla_y H(s)\\-\nabla_x H(s)\end{bmatrix} =\\
    &= \nabla_x\nabla_y H(s) - \nabla_y\nabla_x H(s) = 0.
\end{align}
Note that we are not restricted in the choice of $H(s)$, i.e. it can be completely independent of the target density $p(s)$.
However, if we choose the Hamiltonian as $H(s) = H(x,y) = -p(x)p(y)$, then we obtain the Hybrid Monte Carlo algorithm \citep{duane1987hybrid} as a special case:
\begin{align}
    \frac{\dd x}{\dd t} &= \frac{1}{p(x)p(y)} \nabla_y H(x,y) = -\nabla_y \log p(y),\\
    \frac{\dd y}{\dd t} &= -\frac{1}{p(x)p(y)} \nabla_x H(x,y) = \nabla_x \log p(x),
\end{align}
where $-\log p(x)$ plays the role of the potential energy and the auxiliary kinetic energy is usually taken as $-\log p(y) = \frac{1}{2}||y||^2-\frac{n}{2}\log(2\pi)$, $n$ being the dimensionality of $y$.

\section{Dynamical Gibbs Sampling for Continuous State Spaces}
\label{sec:dgibbs_continuous}

In this section we will describe Dynamical Gibbs sampling via ODEs.
We first formulate dynamical Gibbs sampling in continuous spaces and then generalize to discrete spaces. 

\subsection{1D case}
\label{sec:1d_case}

We start our reasoning with the 1D case.
\begin{lemma}
Any deterministic sampler in 1D is of the form
\begin{align}
    \frac{\dd x}{\dd t}
    =
    v(x) = \frac{c}{p(x)}, \text{ where $c$ is a constant}\,,
    \label{eq:1d_gibbs}
\end{align}
which can be integrated in terms of the CDF $F_p$ as
\begin{align}
    F_p(x(t)) = ct + F_p(x_0)\,.
    \label{eq:1d_gibbs_integrated}
\end{align}
\end{lemma}
\begin{proof}
The first part follows immediately since in 1D theorem \ref{thm:liouville} becomes:
$
    \text{div}(pv) = \frac{\dd}{\dd x} (p(x)v(x)) = 0
$.
The second part can be verified easily by 
taking the time derivative of \eqref{eq:1d_gibbs_integrated}.
\end{proof}
Note, however, that the direction of the velocity field on the line doesn't change the sign since $p(x) > 0$. Thus, for positive $c$, the particle goes to $x\to +\infty$ and never comes back. Moreover, it visits $+\infty$ in a finite amount of time $t_\infty = (1-F_p(x_0))/c$, as one can derive from \eqref{eq:1d_gibbs_integrated} by setting 
$x(t_\infty) = \infty$.
To ensure that the system is able to visit every state we map the dynamics on the circle gluing $+\infty$ to $-\infty$. In practice, this can be achieved, for instance, by mapping the state space via a hyperbolic tangent.
Another way is to use the CDF $F_p$ to map $\mathbb{R}$ onto $[0,1]$, and update the state modulo $1$ in \eqref{eq:1d_gibbs_integrated}:
\begin{align}
    x(t) = g(x_0,t) = F_p^{-1}\bigg(
    (ct + F_p(x_0))\mod 1\bigg).
\end{align}
We now define the dynamic Gibbs sampler in 1D as the procedure that collects values $x(t)$ as we evolve the initial condition according to the ODE \eqref{eq:1d_gibbs}.
We remark that this is related to the procedure described in \citep[Fig. 1]{murray2012driving}, where time evolution now furnishes the stream ${\cal D}$ used in that work.

\subsection{General case}

Given the 1D sampler, Gibbs sampling now allows us to scale this scheme to an arbitrary number of dimensions $n$ assuming that the conditional densities are available.
Recall that the Gibbs algorithm samples one coordinate at a time $x_i \sim p(x_i\cond x_{\setminus i})$, where $x_{\setminus i}$ denotes all the coordinates except $i$-th, and we use the most recent values of $x_{\setminus i}$ in the conditioning statement.
Thus, according to Gibbs sampling, at each iteration, we should update the state along the current dimension as
\begin{align}
    x_i(t+\epsilon) = x_i(t) + \int_t^{t+\epsilon} \dd t'
    \frac{c_i}{p(x_i(t')\cond x_{\setminus i})},
\end{align}
where all values $x_{\setminus i}$ are fixed and we sample only $x_i$.
A significant downside of this approach is that we update only one coordinate at a time increasing the auto-correlation of the samples.
Fortunately, our dynamical formulation allows us to update all of the dimensions simultaneously. We get then to the following procedure.

\begin{definition}[Dynamical Gibbs Sampler]\label{def:dyn_gibbs}
The dynamical Gibbs sampler for continuous state spaces is the deterministic sampler that follows the ODE flow generated by the vector field $v(x)$ with components
\begin{align}
    v_i(x) = \frac{c_i}{p(x_i\cond x_{\setminus i})}.
    \label{eq:dgibbs}
\end{align}
\end{definition}
It is easy to check that it is divergence-free:
\begin{align}
    \text{div}(p(x)v(x)) = \sum_i\deriv{}{x_i} \bigg( p(x)\frac{c_i}{p(x_i\cond x_{\setminus i})}\bigg) = \\ 
    = \sum_i\deriv{}{x_i} \bigg( p(x)\frac{c_i p(x_{\setminus i})}{p(x)}\bigg) = 0\,,
\end{align}
which ensures that the induced dynamics preserves the target density according to theorem \ref{thm:liouville}.

Thus, we can update all of the dimensions at once, and the dynamics adaptively determines the speed of each dimension: the speed is higher for lower conditional density.
Setting $c_i = 1$ for all $i$, we get a version of the algorithm described in \citep{suzuki2013monte}. In that paper it was conjectured that the algorithm actually preserves the target density, which immediately follows from our general theory.  
Thus, in the current paper, for the first time, we prove that the algorithm from \citep{suzuki2013monte} indeed preserves the target measure -- we discuss this in more detail in Section \ref{sec:related}.

\subsection{Ergodicity}

To apply a dynamical system to sampling, we need to cover the state space densely, i.e.~we need the system to be ergodic.
Ergodicity of the Gibbs dynamics depends on the coefficients $c_i$ as we now show.
In 1D any non-zero $c$ in \eqref{eq:1d_gibbs} will allow us to cover the interval $[0,1]$ densely under continuous time dynamics.
For the 2D case, we demonstrate here that by taking rationally independent $c_1$ and $c_2$ we guarantee ergodicity for any target density.
To show that we first need the following lemma.
\begin{lemma}\label{lemma:phase_curve}
Define the phase curve of a trajectory $x(t)$ as the curve that the trajectory draws on the manifold on which $x(t)$ lives.
For any continuous vector field $v(x) \in \mathbb{R}^n$ and positive scalar function $s(x)$, the solutions of the following ODEs
\begin{align}
    \frac{\dd x}{\dd t} = v(x), \;\mathrm{ and }\; \frac{\dd x}{\dd t} = s(x)v(x), \;\; s(x) > 0 \;\; \forall x
\end{align}
have the same phase curves.
\end{lemma}
Rephrasing, drawing the phase curves for both systems and forgetting about time, we cannot distinguish these solutions (see proof in Appendix \ref{app:scaling_lemma}).
We can now prove the ergodicity condition for 2D.

\begin{lemma}[2D ergodicity]\label{lemma:2d_ergodicity}
The dynamical Gibbs sampler of Def.~\ref{def:dyn_gibbs} is ergodic in 2D if and only if $c_1$ and $c_2$ are rationally independent.
\end{lemma}
\begin{proof}
From lemma \ref{lemma:phase_curve} we know that
to study ergodicity we can multiply equation \eqref{eq:dgibbs} by $p(x)$ and not change the phase-curves. In 2D, this yields
\begin{align}
    \frac{\dd x_1}{\dd t} 
    = c_1
    p(x_2),
    \quad
    \frac{\dd x_2}{\dd t} = c_2p(x_1)
    \,.
\end{align}
We then transform to new coordinates $(u_1,u_2) = (F_1(x_1), F_2(x_2))$, where $F_i$ is the CDF of the $i$-th marginal.
The dynamical equations transform as
\begin{align}
    \frac{\dd u_1}{\dd t} 
    = c_1p(x_1)p(x_2),\quad 
    \frac{\dd u_2}{\dd t} 
    = c_2p(x_1)p(x_2)
    \,.
\end{align}
Applying the lemma \ref{lemma:phase_curve} once again  to divide the vector field by $p(x_1)p(x_2)$, we obtain a rotation on the torus with speeds $c_1$ and $c_2$, which is dense when $c_1$ and $c_2$ are rationally independent \cite{fomin1981ergodic}.
\end{proof}

We have the following more general result:
\begin{lemma}[Ergodicity for factorized marginal]
If any marginal $p(x_{\setminus i})$ factorizes as a product of distributions along each dimension ($n-1$-wise independent distributions), 
all coefficients $c_i$ must be rationally independent for ergodicity to hold.
\end{lemma}
\begin{proof}
The same technique used in lemma \ref{lemma:2d_ergodicity} can be applied in this case to see that we also map to a rotation on the $n$-dimensional torus where rational independence needs to hold for ergodicity.
\end{proof}

Unfortunately, the technique used to prove these lemmas cannot be applied to arbitrary high-dimensional distributions. We conjecture that in general ergodicity does not depend on the target density and depends only on the coefficients $c_i$.
That is why, in implementations of the algorithm, we set all of the coefficients as square roots of prime numbers (see \citep{besicovitch1940linear} for the proof of their independence).

\section{Dynamical Gibbs for Discrete State Spaces}

\begin{figure}
    \centering
    \includegraphics[width=0.23\textwidth]{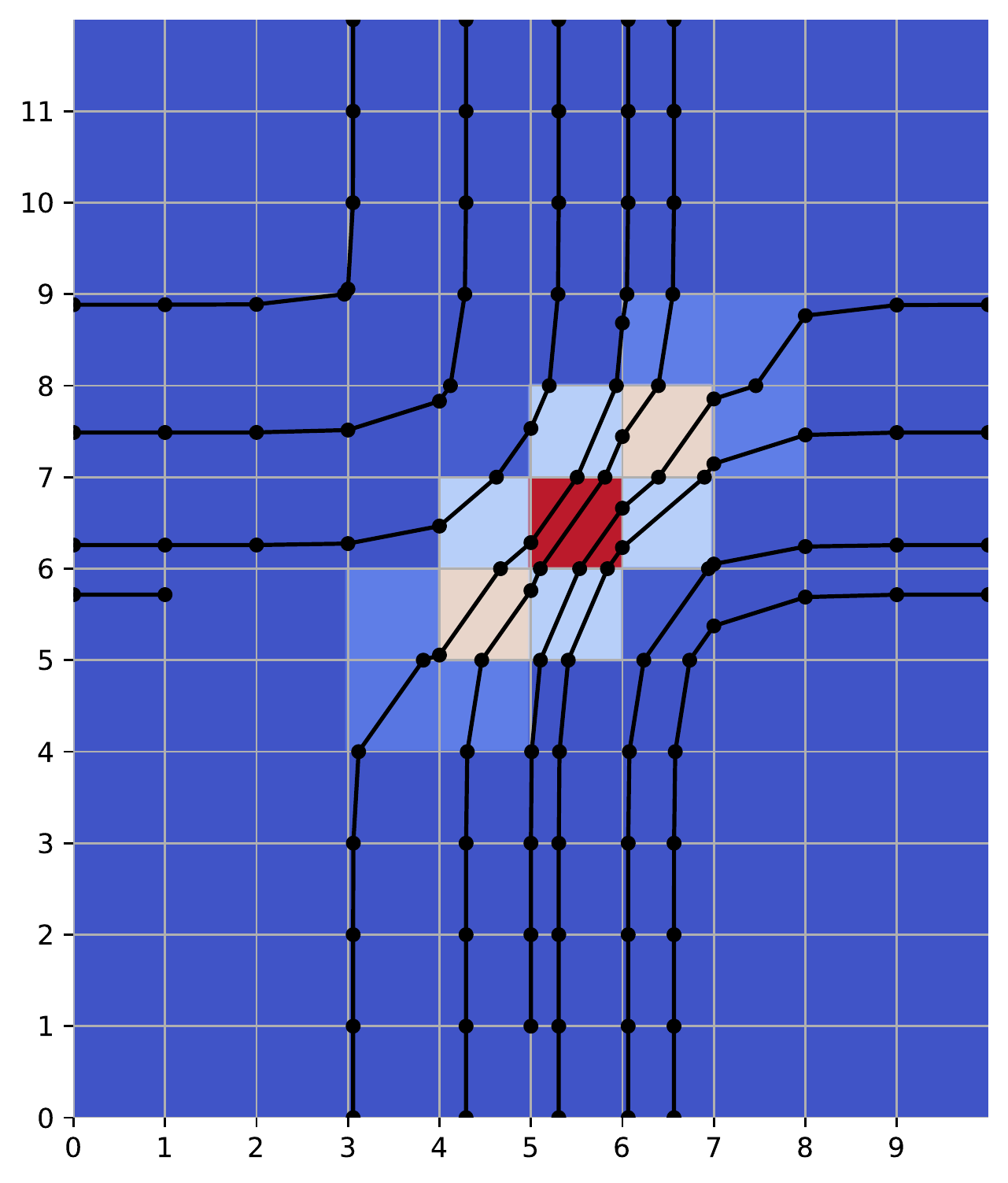}
    \includegraphics[width=0.23\textwidth]{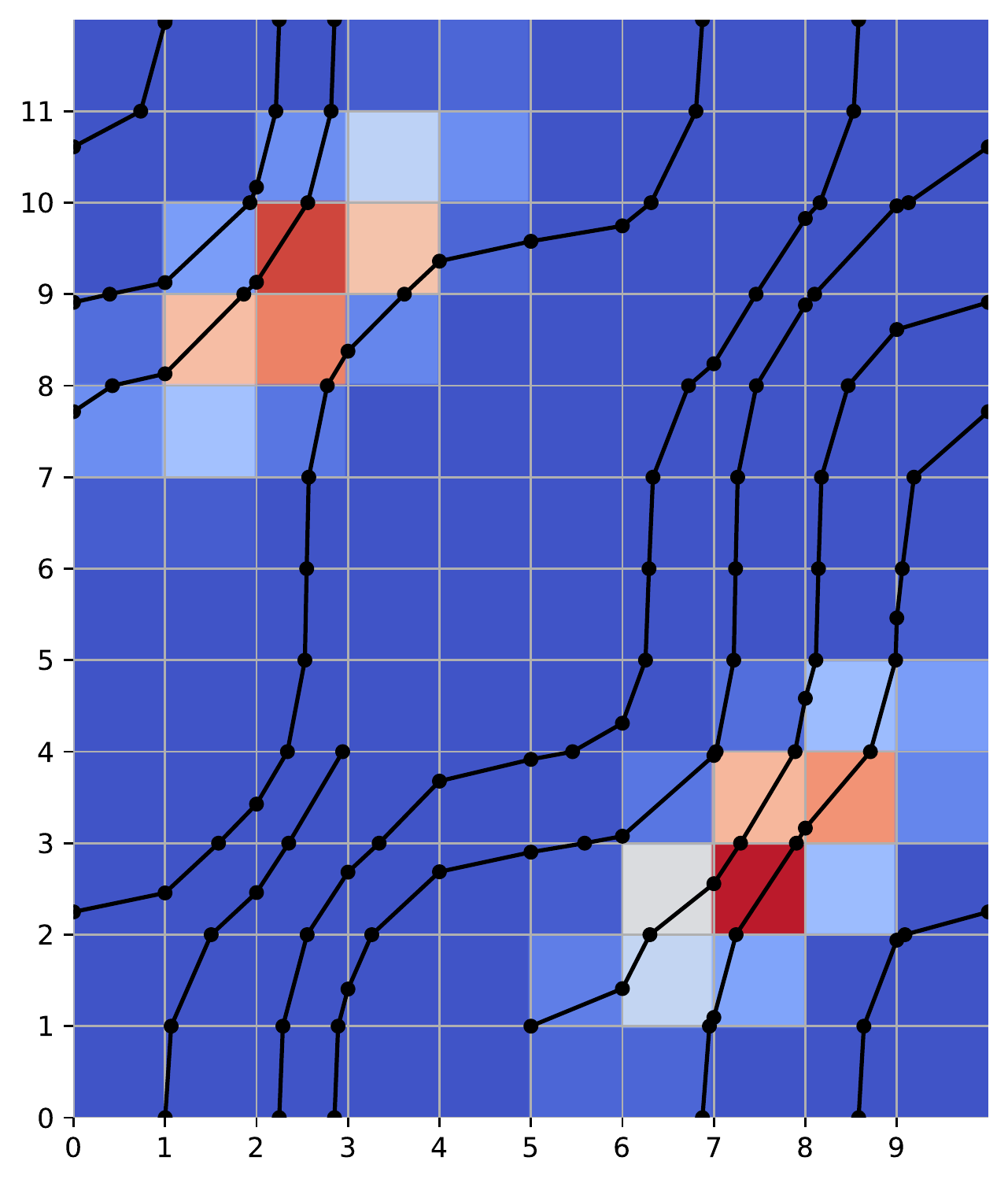}
    \caption{Illustrations for the discrete state space. Red cells correspond to higher probability regions. Black dots correspond to the change points and black lines correpond to the trajectory inside each cell.}
    \label{fig:2d_illustrations}
\end{figure}

Another merit of the proposed scheme is that it does not use any information about gradients of the target density.
In this section we demonstrate that this property can be leveraged for sampling from discrete random variables.

We first map the discrete distribution to the continuous domain as follows (known as ``dequantization'').
We assume that each discrete variable $x_i$ takes integer values $[0,1,\ldots,d_i-1]$.
Then we set the density for all points in the hypercube $\{x+u: u\in[0,1]^n\}$ to $p(x)$.
Thus, we split the whole space into cells of constant density in the continuous domain.
Clearly, sampling from such distribution is equivalent to sampling from the corresponding discrete distribution: to obtain samples from the original discrete distribution we can simply take $\lfloor x\rfloor$ for all collected ``continuous'' samples.
Note, however, that the Liouville's theorem cannot be applied straightforwardly to this discretized case since the vector field is not continuous anymore.
In Appendix \ref{app:disc_liouville} we provide a proof of Liouville's theorem for this discontinuous case.

The discrete case has two practical advantages compared to the continuous case.
Firstly, we don't need to apply any transformation to obtain the distribution on the $n$-dimensional torus as we discussed in Section \ref{sec:dgibbs_continuous}.
The torus is naturally obtained just by considering all the variables $x_i$ modulo $d_i$.
Secondly, we can exactly integrate the dynamics inside each cell since the velocity inside each cell is constant, i.e. the particle moves in a straight line, and we simply jump between the boundaries of the cells. Since the whole trajectory inside each cell corresponds to a single discrete state $\lfloor x\rfloor$, we simply weight each sample by the time $dt$ that it spent inside that cell, equal to the length of the trajectory divided by the speed. This procedure of generating an importance weighted set of samples is efficient and largely avoids the accumulation of numerical integration errors. 

Denoting the discrete sample as $x[k]$ and the time spent in it as $t_k$, the mean of function $f$ can be estimated as
\begin{align}
    \mathbb{E}_{p(x)}f(x) \simeq \frac{1}{\sum_k t_k}\sum_k t_k f(x[k]).
\end{align}
Finally, to provide the reader with more intuition we demonstrate the resulting dynamics in 2D case (see Fig. \ref{fig:2d_illustrations}).


\section[Classification of Deterministic Samplers]{Classification of Deterministic Samplers\footnote{This section assumes background in differential geometry, but can be skipped without significant loss for the exposition of the proposed algorithm.}}
\label{sec:hodge}

As Liouville theorem (Thm.~\ref{thm:liouville}) shows, the only thing one needs to design a measure-preserving flow is to design a divergence free vector field.
Indeed, given the vector field $w: \text{div}(w) = 0$ we can easily obtain the measure-preserving vector field $v(x) = w(x)/p(x)$ for which we have $\text{div}(pv) = 0$.
This fact leads us to the question: how can we describe the class of the vector fields $\{w: \text{div}(w) = 0\}$. This question can be approached via particular Hodge theory. 

\subsection{Background}

To state the results, we first give some basic definitions, see \cite{vogtmann1997mathematical,nakahara2003geometry} for more background. Given a vector space $V$, an exterior $k$-form is an antisymmetric linear map from $V^{\times k}$ to $\mathbb{R}$. For example, in $V=\mathbb{R}^2$, a $2$-form is the oriented area of the parallegram constructed from two vectors.
In terms of the basis $e_i$, $i=1,\dots,n$ of $V$, a $k$-form is $\omega = \sum a_{i_1,\dots,i_k} e_{i_1} \wedge \cdots \wedge e_{i_k}$, with 
$\wedge$ the exterior product, $a\wedge b = -b\wedge a$ and the summation is over $i_1<\cdots <i_k$.
The Hodge star operator \cite{nakahara2003geometry} \footnote{The reader familiar with quantum physics might appreciate that $k$-forms describe states of fermions, the exterior product implementing Pauli exclusion principle, and $\star$ is a particle-hole transformation.} maps $k$-forms to $(n-k)$-forms:
\begin{align}
\star 
e_{i_1} \wedge \cdots \wedge e_{i_k} = 
(-1)^{\sigma} 
e_{i_1'} \wedge \cdots \wedge e_{i_{n-k}'}
\end{align}
with $\{i'_j\}$ being the complementary set to $\{i_j\}$ and $\sigma$ the sign of the permutation $(i_1,\dots,i_k,i'_1,\dots,i'_{n-k})$.

For a manifold $M$ we define a differential $k$-form at a point $x\in M$ as an exterior form on the tangent space to the manifold $TM_x$.
So in local coordinates, a differential $k$-form reads
$\alpha(x) = \sum a_{i_1,\dots,i_k}(x) \dd x_{i_1} \wedge \cdots \wedge \dd x_{i_k} $,
where $\dd x_i$ is a basis of the dual $(TM_x)^*$. Denote $\Omega^k(M)$ the space of $M$ forms.
The exterior derivative is $\dd : \Omega^k(M) \to \Omega^{k+1}(M)$. It is defined for $f\in \Omega^0(M)$ (i.e.~a function on the manifold) as the differential $\dd f (x) = \sum_i \partial_i f \dd x_i$ and extended to general forms as
$
\dd \alpha = 
\sum \dd a_{i_1,\dots,i_k}\wedge \dd x_{i_1} \wedge \cdots \wedge \dd x_{i_k}
$. Crucially, $\dd^2 = 0$. 
$\Omega^1(M)$ can be identified with the space of vector fields under the isomorphism $\sharp$ which associates covectors to vectors. For example, $(\dd f)^\sharp = \nabla f$. The inverse isomorphism is denoted $\flat$.
The codifferential $\delta = \star \dd \star$ plays the role of the divergence:
given $\omega=\sum v_i \dd x_i \in\Omega^1(M)$ in local coordinates, we have
$\delta\omega = \text{div}(v)$ as can be easily verified.

\subsection{Example: divergence-less fields in $\mathbb{R}^3$}
This formal machinery pays off since it gives an economical and general description of divergeless vector fields. To illustrate the formalism and motivate the following developments, we rederive the result $\text{div} (\text{curl}(F) ) = 0$ for $F$ a vector field in $M=\mathbb{R}^3$. (Recall that $\text{curl}(F)_i=\epsilon_{ijk}\partial_j F_k$ using the Einstein convention to sum over repeated indices.)
The following statement can be easily verified. (See Appendix \ref{app:hodge} for the proof.)

\begin{lemma}
Given the vector field $F = (F_1,F_2,F_3)$ and $\alpha = v^\sharp = \sum_{i=1}^3 F_i \dd x_i$, we have ${\rm curl}(F)^\sharp = \star \dd \alpha$.
\end{lemma}
This allows us to show that $\delta \star \dd \alpha = \star \dd^2 \alpha = 0$ proving the classical result $\text{div} (\text{curl}(F) ) = 0$. The converse result, that any divergence-less field in $\mathbb{R}^3$ can be written as $v=\text{curl}(F)$, is known as Helmholtz's theorem. In fact, one can show that all divergence-less fields are of this form for the more general case of $\mathbb{R}^n$ as we discuss next.

\subsection{The case of $\mathbf{\mathbb{R}^n}$ (Poincare lemma)}

As already remarked we can work with one forms instead of vector fields and the divergence of $\omega \in \Omega^{1}(\mathbb{R}^n)$ is $\delta\omega = \star \dd \star \omega$.
Hence, the set of divergenceless forms is defined as $\delta\omega = 0$.
Clearly if there exists $\alpha \in \Omega^{n-2}(\mathbb{R}^n)$ such that $\star \omega = \dd \alpha\in \Omega^{n-1}$, then $\delta\omega  = 0$ since $\dd^2=0$.
Poincar\'e lemma \cite{vogtmann1997mathematical} states that the converse is also true: any divergenceless field can be written as $\omega = \star \dd \alpha$ or $\omega=\delta \beta$, for a given 
$\beta \in \Omega^{2}(\mathbb{R}^n)$. In terms of vector fields of our setting:
\begin{align}
    v = \frac{1}{p}(\delta \beta)^\sharp \,.
\end{align}
Due to antisymmetry, the dimensionality of $2$-forms in $\mathbb{R}^n$ is $n!/(2 (n-2)!)$, i.e. any measure preserving flow can be defined by $n!/(2 (n-2)!)$ continuous functions.

\subsection{The case of compact manifolds (Hodge decomposition)}

When the state space is given by a compact manifold $M$, we can use the Hodge decomposition instead. It states that $\omega \in \Omega^{1}(M)$ such that $\delta \omega=0$ can be written as $\gamma + \delta\beta$, with $\gamma\in\Omega^1(M)$ harmonic, i.e.~$\dd \gamma = \delta\gamma = 0$. We do not give further justification here for this result and refer the interested reader to \cite{nakahara2003geometry}.
Therefore one can design the vector field $v$ preserving $p$ as
\begin{align}
    v = \frac{1}{p}(\gamma + \delta\beta)^\sharp,
    \label{eq:v_hodge}
\end{align}
where we take any $\beta \in \Omega^2(M)$ and the harmonic form $\gamma \in \Omega^1(M)$.
Moreover, this representation of the sampler is unique.
As in the HMC example of Sec.~\ref{sec:hmc}, the family of samplers is independent of $p$, i.e. we can use the vector field $(\gamma + \delta\beta)^\sharp$ to sample from any density just by multiplying by $1/p$.
However, the efficiency of the sampler depends on the field we have chosen.

Compared to the case of $\mathbb{R}^n$ we have another component here --- harmonic vector fields.
The dimensionality of this component $\text{ker}_{\Delta}(\Omega^1(M))$ is given by the corresponding Betti number $\beta_1(M)$ \cite{nakahara2003geometry}.
For instance, for the $n$-dimensional torus $S^n$, $\beta_1(M) = n$.

\textbf{Decomposition of a sampler.}

The Hodge decomposition allows us not only to design samplers, but decompose existent.
Indeed, given the $1$-form $\omega$ we can find its components by solving the following optimization problem
\begin{align}
  \beta = \argmin_{\widehat{\beta}} \norm{(vp)^\flat-\delta\widehat{\beta}},  
\end{align}
then the harmonic component is given by $\gamma = (vp)^\flat-\delta\beta$.
For the characterization of dynamical Gibbs see Appendix \ref{app:hodge}.
We leave the study harmonic samplers as an interesting future direction for our work.


\section{Related Work}
\label{sec:related}

The works most related to ours are \citep{murray2012driving},
\cite{neal2012view}, and \cite{suzuki2013chaotic}.

\citep{murray2012driving} integrates the equations of motion in 1D as discussed at the end of \ref{sec:1d_case} and uses the resulting update scheme as plug and play in several MCMC samplers. 
Our formulation in terms of differential equations has one important benefit – we don't need to evaluate the CDF and its inverse anymore. 

\cite{suzuki2013chaotic} investigates deterministic samplers based on ODEs for discrete problems. Our work is closely related, but we improve on 
\cite{suzuki2013chaotic} in several directions: 1) we prove that the sampler preserves the target distribution; 2) we introduce coefficients $c_i$ that are required for ergodicity; 3) we give a general classification of all deterministic samplers.
In the experimental section and in the Appendix \ref{app:experiments} we provide the empirical evidence that the algorithm of \cite{suzuki2013chaotic} is not ergodic.

\section[Experiments]{Experiments\footnote{code reproducing all experiments is available at \href{https://github.com/necludov/continuous-gibbs}{github.com/necludov/continuous-gibbs}}}
\label{sec:experiments}

In our experiments, we focused on discrete random variables since it allows us to estimate the performance of the proposed algorithm without the introduction of numerical errors, and any accept-reject tests to alleviate these errors.
We provide the code for all experiments in the supplementary material.

\subsection{2D discrete distribution}

We first illustrate the performance of our algorithm on 2D discrete random variable. 
For the target distribution we take the grayscale portrait of Gibbs (Fig. \ref{fig:portrait_asymptotics} right) and treat each pixel value as the probability of its coordinates to appear.
The resolution of the image is $237\times 178$, thus we have the joint distribution of two discrete random variables: the row number $\in [0,\ldots,236]$ and the column number $\in[0, \ldots, 177]$.
For the convergence speed, we evaluate the true mean values of the coordinates and then evaluate the error of its estimation with the number of samples.
Each sample corresponds to one update of a dimension for our algorithm and for the Gibbs sampler, while the independent sampler updates both coordinates at a time.

\begin{figure}[t]
    \centering
    \includegraphics[width=0.27\textwidth]{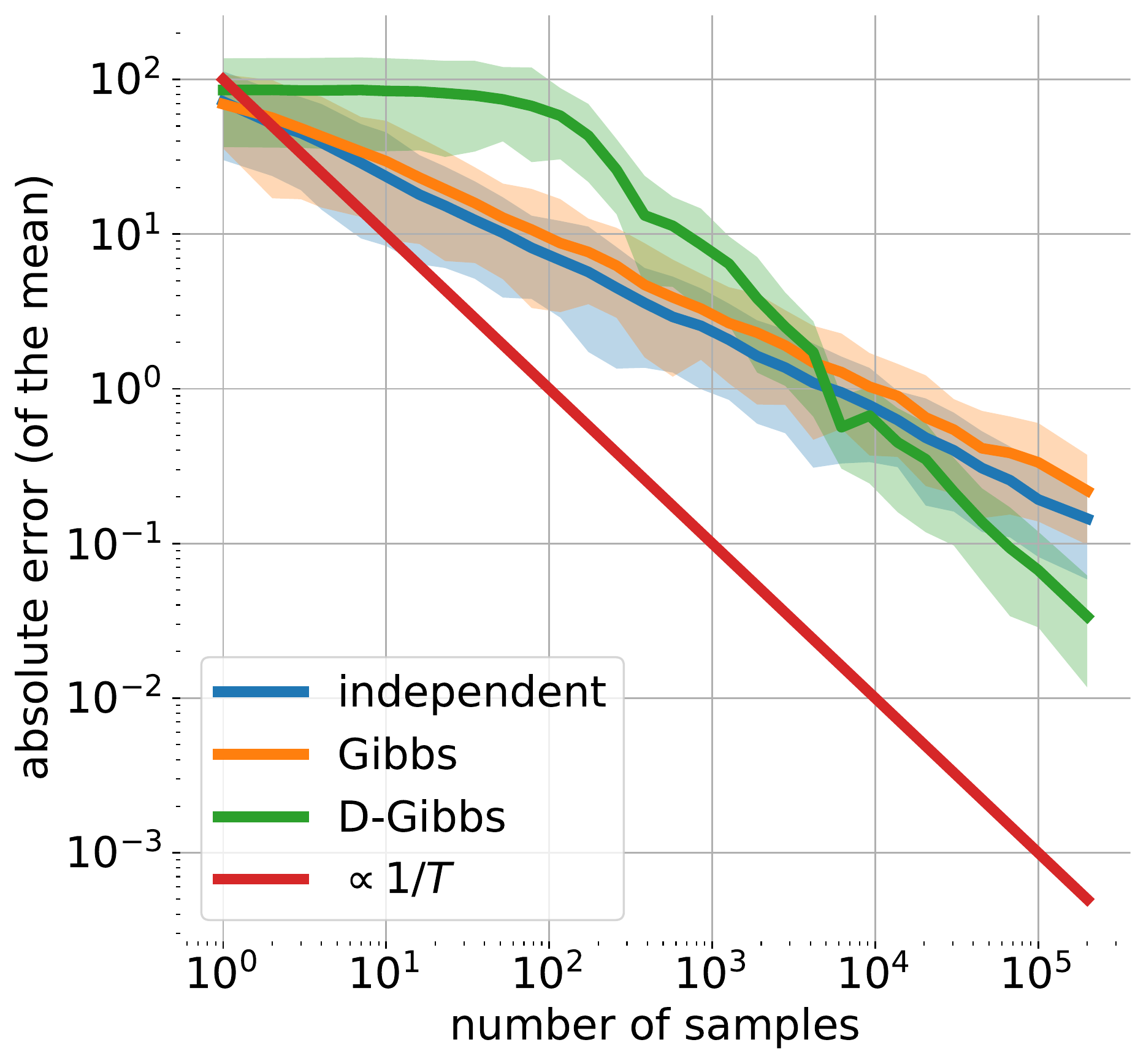}
    \includegraphics[width=0.18\textwidth]{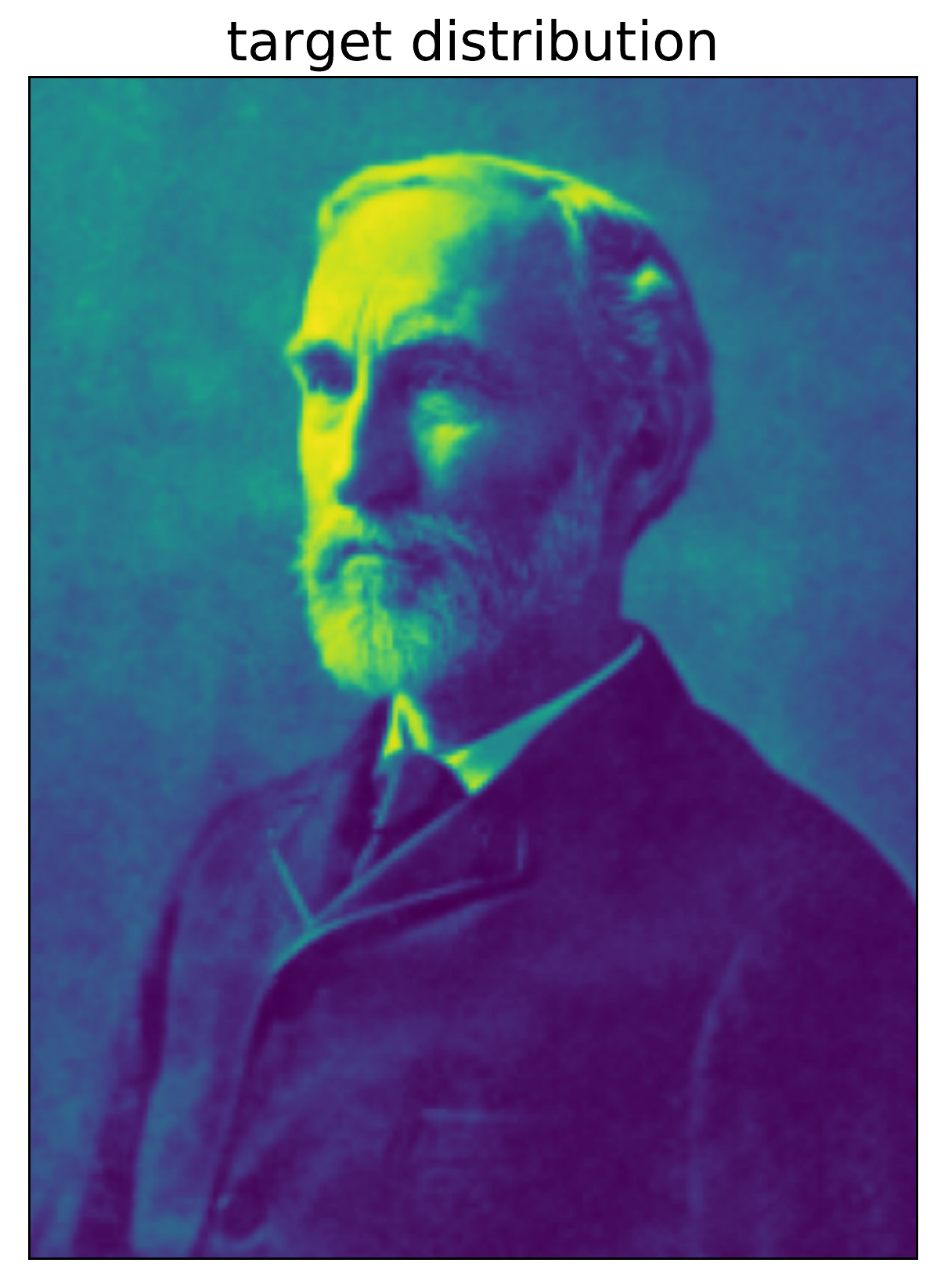}
    \caption{Asymptotics of the convergence of the estimated mean to the true mean of the 2D distribution (on the right), averaged across $10^3$ independent runs. The shaded area depicts $0.1$ and $0.9$ quantiles.}
    \label{fig:portrait_asymptotics}
\end{figure}

In Fig. \ref{fig:portrait_asymptotics}, we compare the proposed algorithm (D-Gibbs) against independent sampling and Gibbs algorithm.
The main takeaway of this experiment is that our algorithm has better asymptotic performance ($1/T$ convergence) even than the independent sampler ($1/\sqrt{T}$ convergence).
This result is evident even on the 2D histograms of samples (Fig. \ref{fig:portrait} in the introduction).
We provide the histograms for Gibbs sampling and for the algorithm proposed in \citep{suzuki2013monte} in Appendix \ref{app:experiments_portrait}.
There one can see that Suzuki's algorithm is not ergodic for 2D, as we prove in Lemma \ref{lemma:2d_ergodicity}.

\subsection{Ising model}

In the next experiment we consider the Ising model on the square lattice.
That is, each point of the square lattice ($28\times28$) is associated with the spin value $\sigma_i \in \{-1,+1\}$.
The total energy of the system is evaluated as 
\begin{align}
    E(\sigma) = \sum_{(i,j) \in \text{lattice}} \sigma_i\sigma_j,
\end{align}
where the interactions of spins are evaluated only for adjacent nodes of the square lattice. The unnormalized density of the Boltzmann distribution is then given by $p(\sigma) \propto \exp(-E(\sigma))$. Thus, the low energy values are obtained for checkerboard patterns of the spins, i.e. when adjacent spins have the opposite values.

In Fig. \ref{fig:ising}, we report the convergence speed for three algorithms: the Gibbs sampling, dHMC \citep{nishimura2020discontinuous}, and the dynamical Gibbs (ours).
In practice, samplers usually converge to the one of the global modes (a checkerboard) and stay in its vicinity.
To compute the error of the mean estimation, we check to which mode the sampler has converged and set this value as the true mean value.
The dynamical Gibbs outperforms the other sampling algorithms by converging faster to states with high probability and estimating the mean value more precisely.
The running time of all algorithms is approximately the same, however, to take the computational efforts into account we provide the same plots but in terms of running time (see Appendix \ref{app:experiments_ising}).

\begin{figure}[t]
    \centering
    \includegraphics[width=0.23\textwidth]{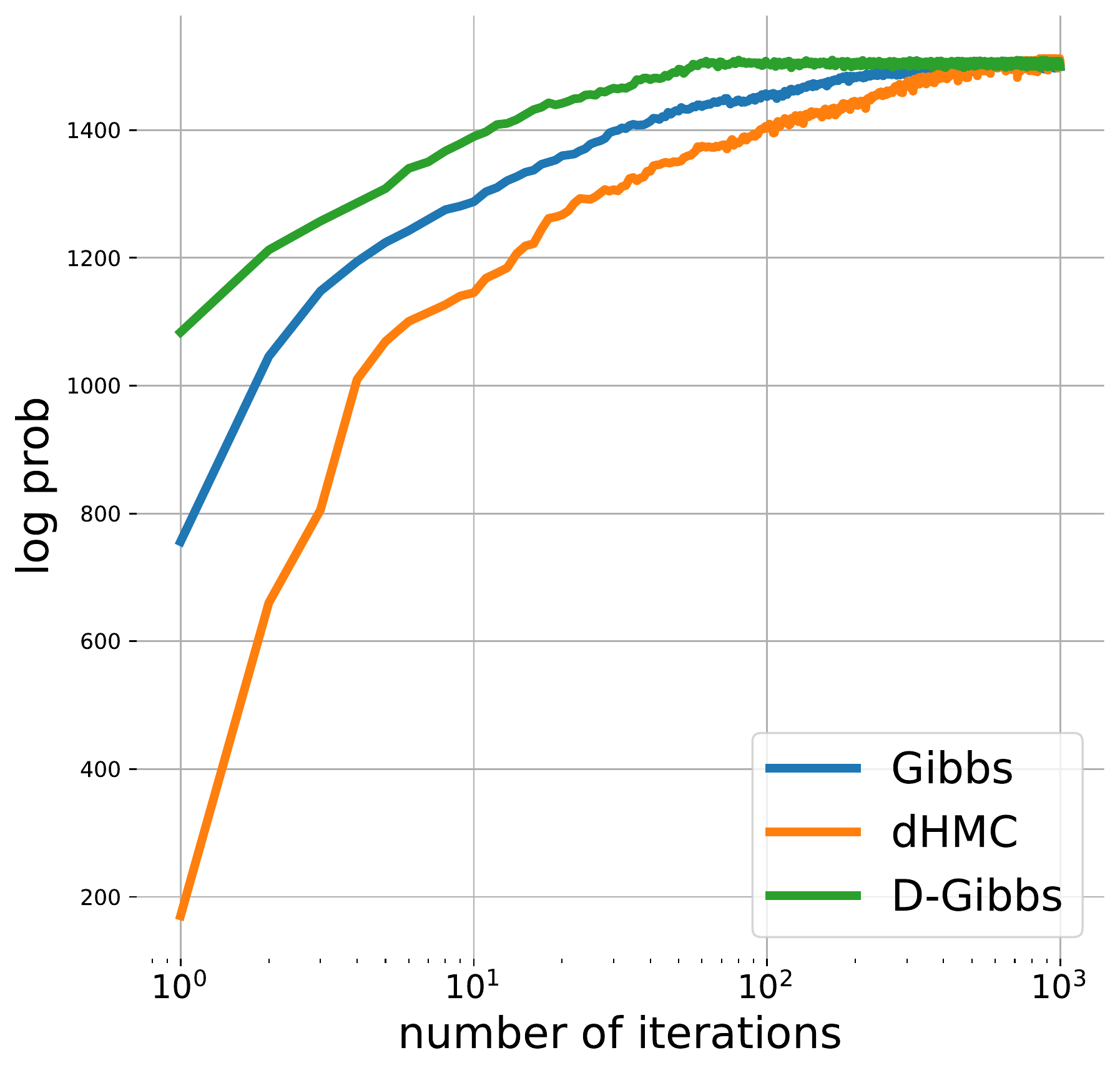}
    \includegraphics[width=0.23\textwidth]{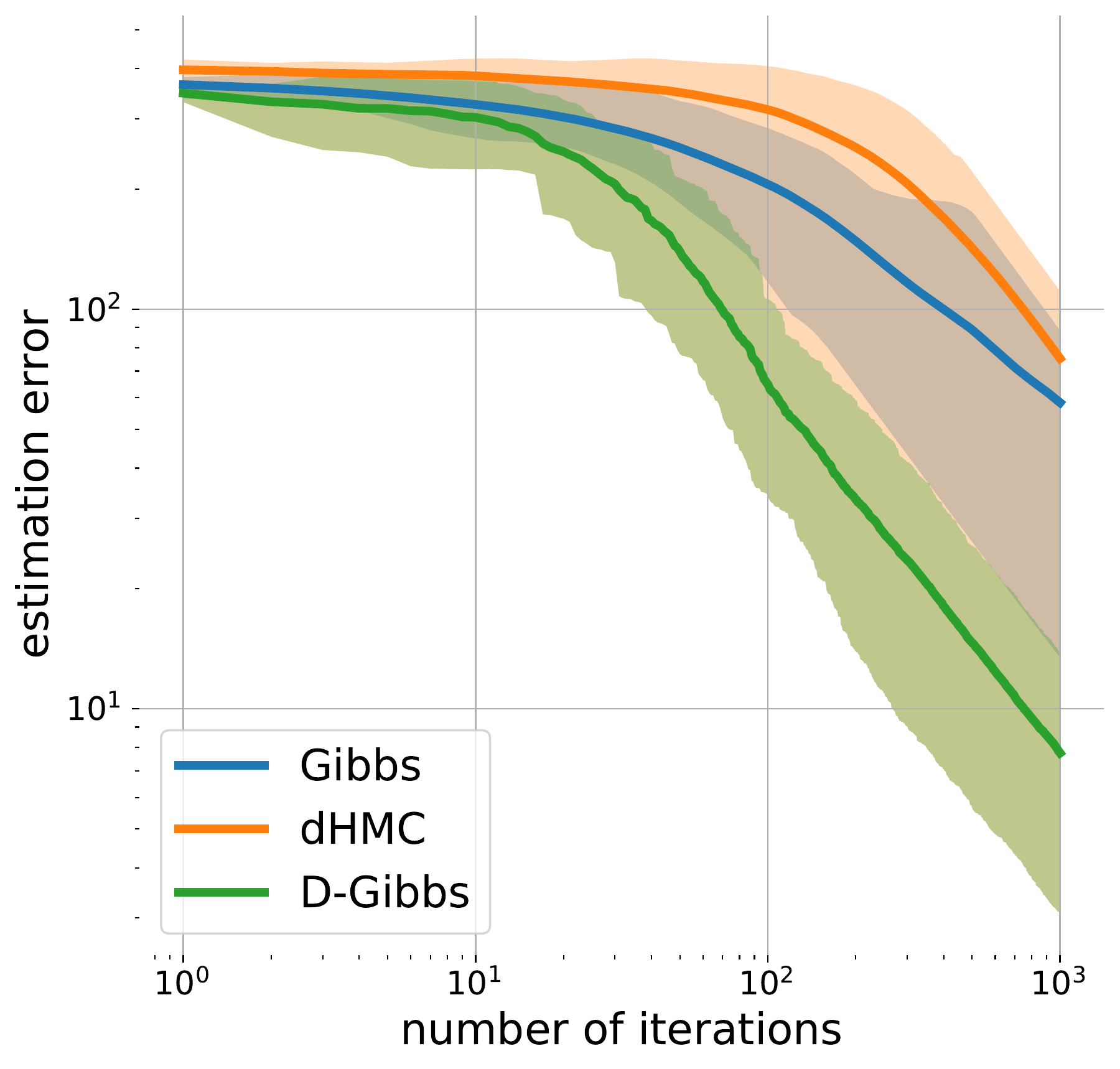}
    \caption{Performance of the samplers for the Ising model. Solid lines demonstrate the mean averaged across $10$ independent runs.}
    \label{fig:ising}
\vspace{-10pt}
\end{figure}

The fast convergence to high probability states motivates another application of the proposed algorithm --- image denoising.
For this purpose, we employ the model from \citep{bishop2006pattern}, which is also the Ising model, but with a slightly different energy:
\begin{align}
    E(\sigma) = -\beta\sum_{(i,j) \in \text{lattice}} \sigma_i\sigma_j -\eta \sum_i \sigma_i \xi_i,
\end{align}
where $\sigma_i \in \{-1,+1\}$ are the values of the state (spins configuration), $\xi_i \in \{-1,+1\}$ are the pixel values of the binarized image, and $\beta=1.0,\eta=2.1$ are the hyperparameters of the model.
Unlike the previous example, the first term here favours the adjacent spins to have the same value and plays the role of smoothing, while the second favours the spins to have the same values as pixels to save the content of the image.

\begin{figure*}[t]
    \centering
    \includegraphics[width=0.95\textwidth]{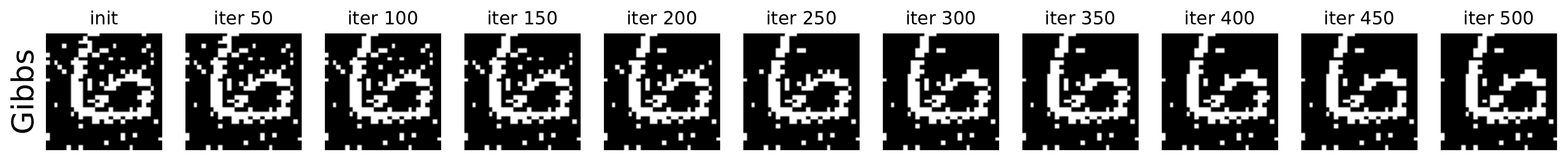}
    \includegraphics[width=0.95\textwidth]{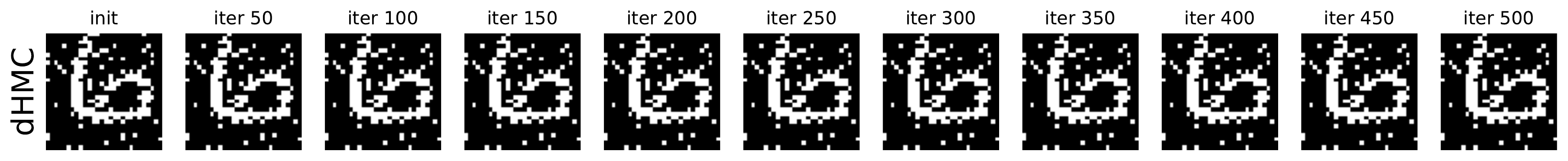}
    \includegraphics[width=0.95\textwidth]{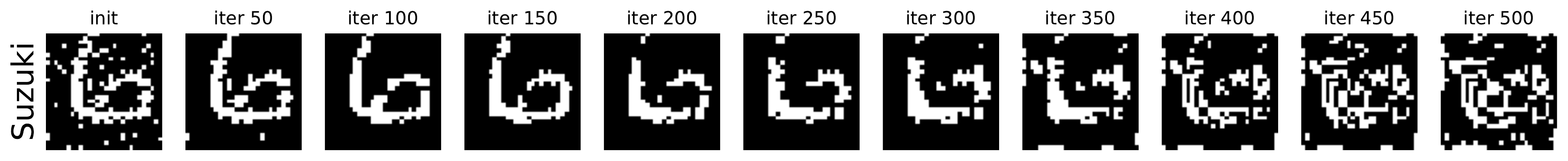}
    \includegraphics[width=0.95\textwidth]{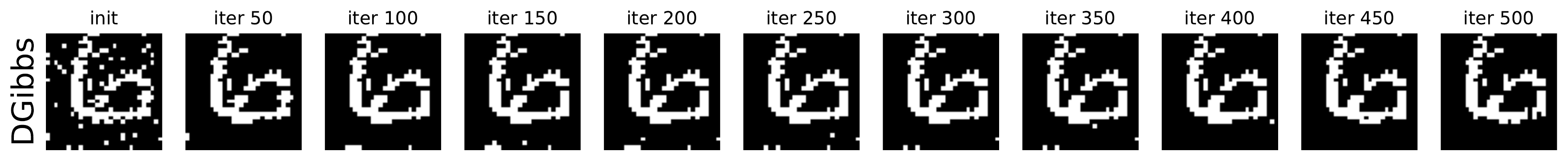}
    \caption{Image denoising via the Ising model. Every iteration correponds to a single update of dimension for all algorithms.}
    \label{fig:denoising}
\end{figure*}
\begin{figure*}[t]
    \centering
    \includegraphics[width=0.23\textwidth]{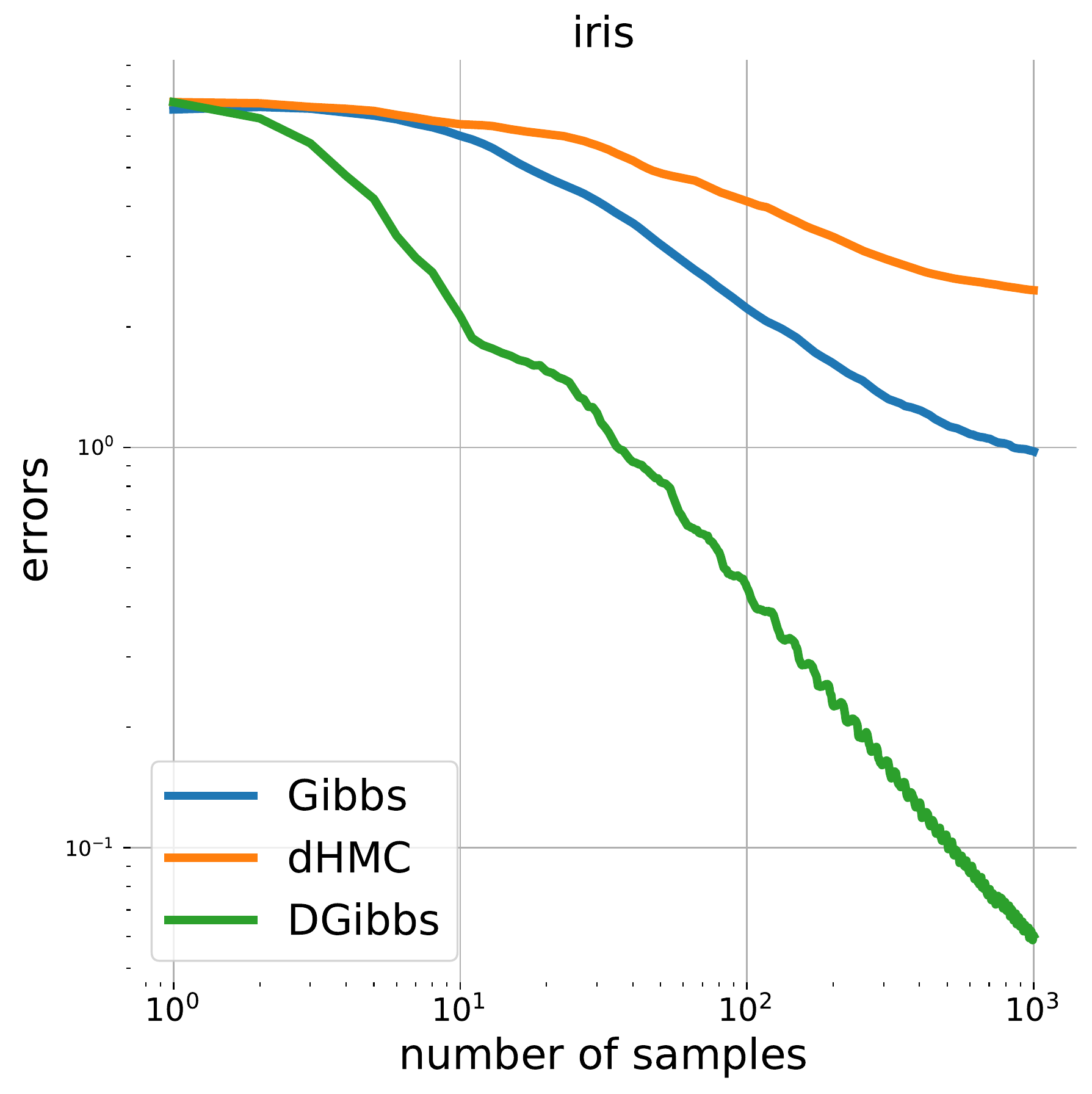}
    \includegraphics[width=0.23\textwidth]{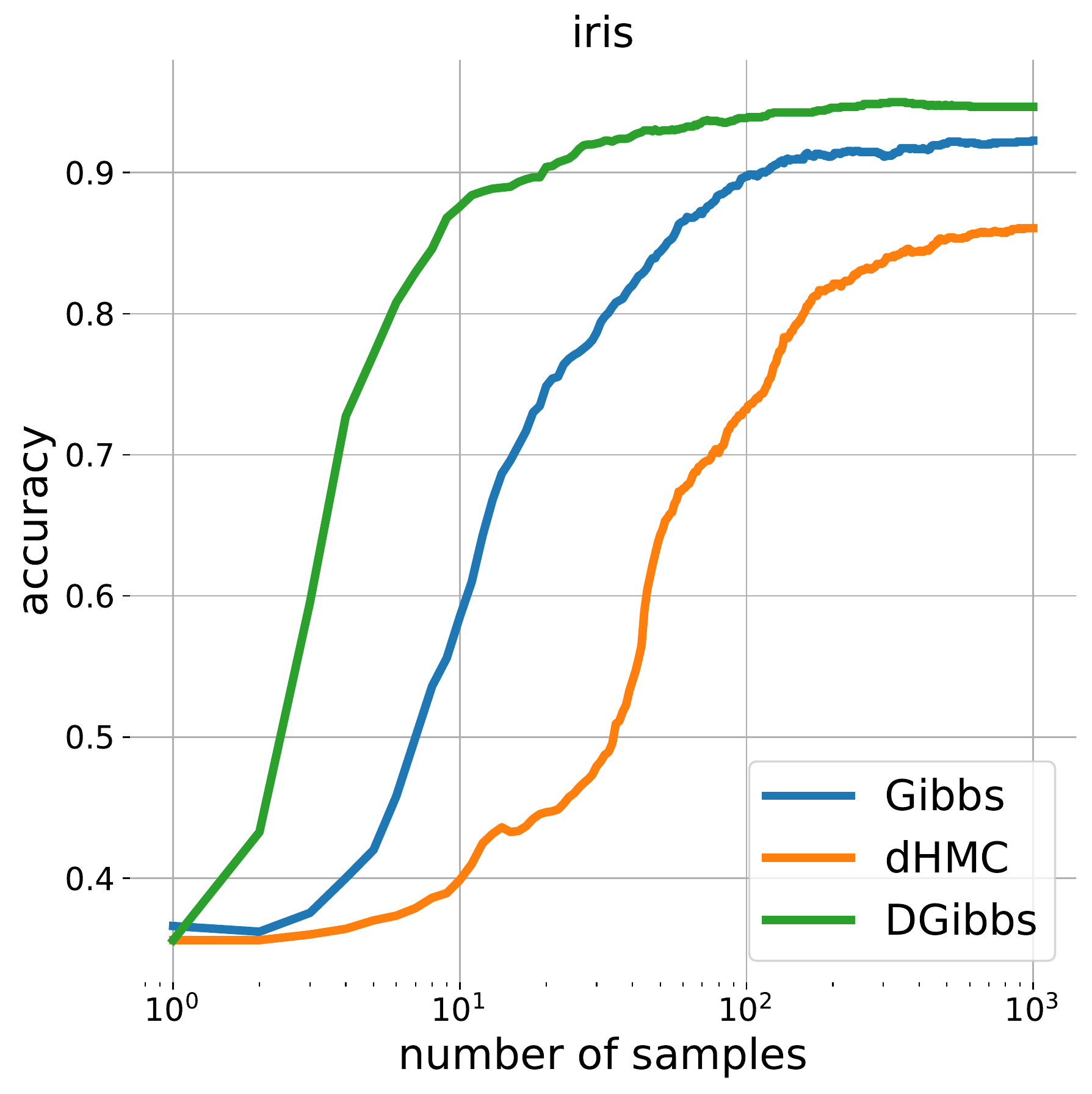}
    \includegraphics[width=0.23\textwidth]{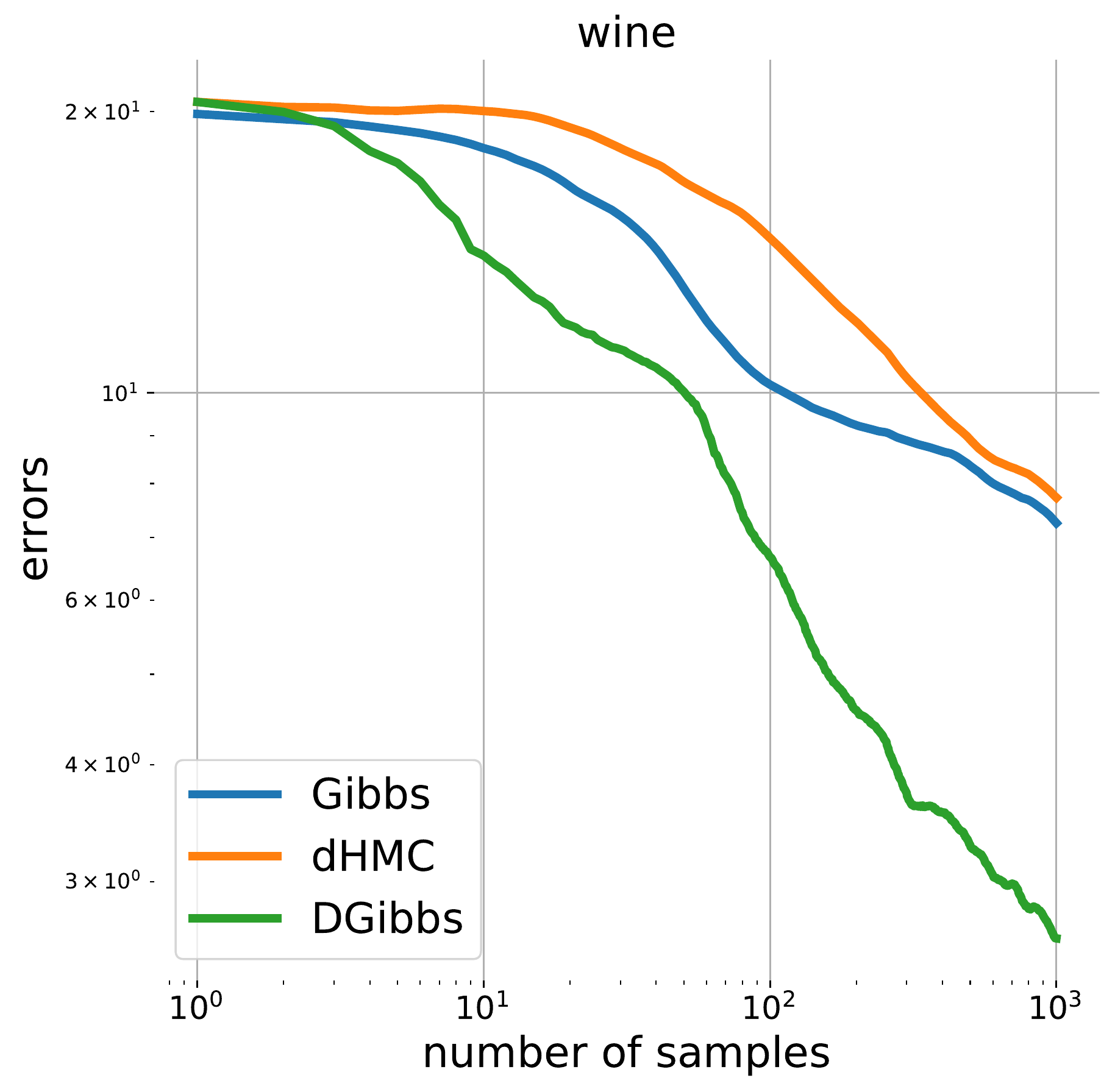}
    \includegraphics[width=0.23\textwidth]{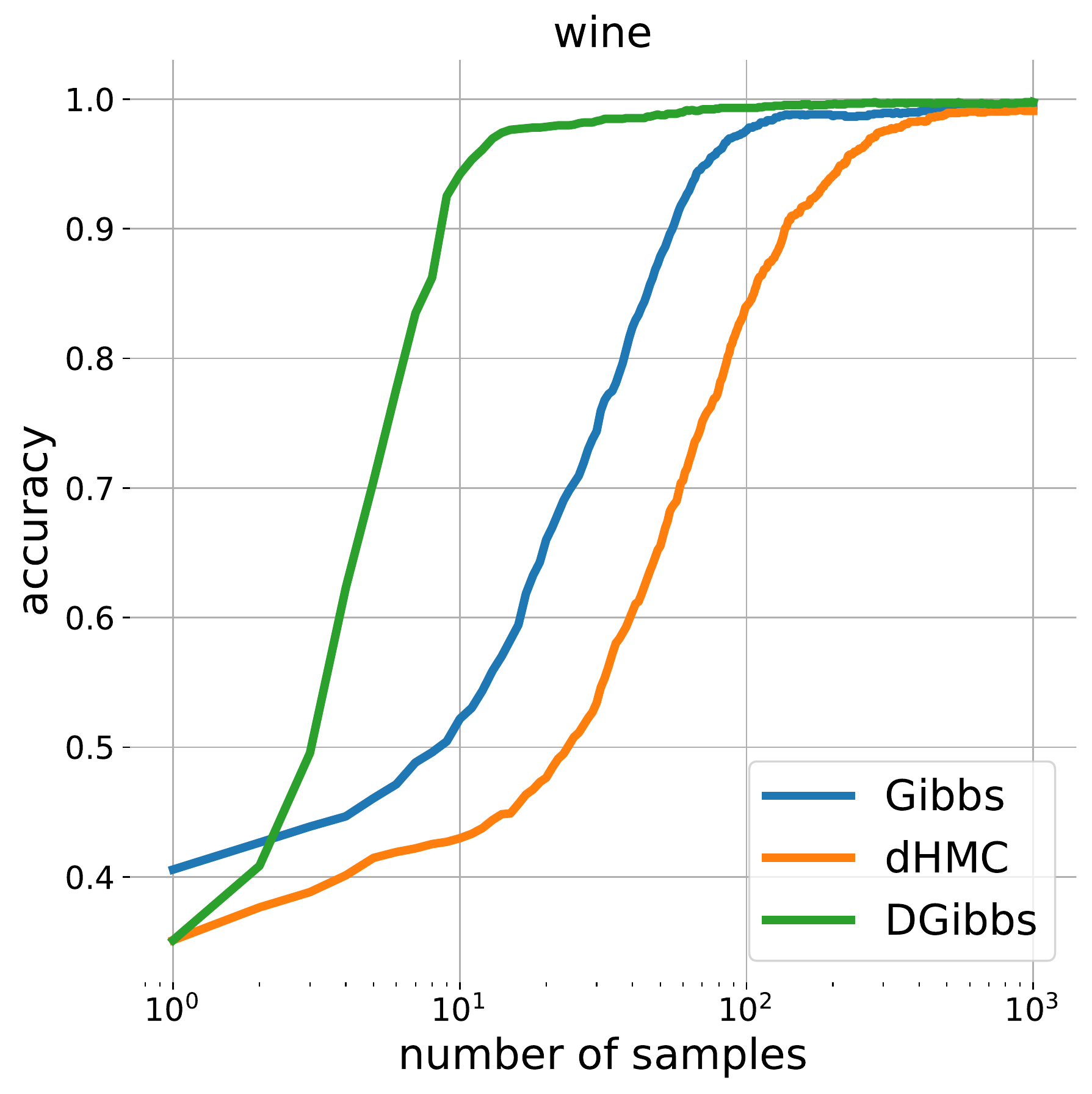}
    \caption{Sampling from the posterior distribution of logistic regression with binarized parameters. As datasets we consider Iris (on the left) and Wine (on the right). For both datasets we report the accuracy of classification estimated by averaging across collected samples and the error in the estimation of mean parameters values. We do not provide quantiles to keep the plots clear.}
    \label{fig:logreg}
\end{figure*}

In Fig. \ref{fig:denoising}, we demonstrate the results for the denoising.
Each iteration corresponds to one update of a dimension.
The dynamical Gibbs converges faster to the high probability states corresponding to the denoised image.
The intuition behind this phenomena is that dynamical Gibbs, as well as Suzuki's algorithm, adaptively set the velocities for each dimension, thus, flipping the noisy pixels first.
At the same time, Gibbs and dHMC need to go over all dimensions sequentially to get to the high probability state.
We also see that Suzuki's algorithm, after convergence to the denoised image, goes away to the low probability states. 
We conjecture that Suzuki's scheme is not ergodic in high dimensions and we observe some periodicity in its dynamics.

\subsection{Binarized logistic regression}

Another application of sampling discrete variables is Bayesian inference for machine learning models.
Here we consider the posterior distribution of logistic regression with binary weights.
For the dataset $\mathcal{D}$, the unnormalized density of the posterior is defined as $p(\theta\cond \mathcal{D}) \propto p(\mathcal{D}\cond\theta)p(\theta)$, where each parameter $\theta_i$ takes binary values $\{-1,+1\}$.
For the prior we take the uniform distribution, hence, the density is completely defined by the likelihood.

To estimate the performance of samplers, we report two metrics: the accuracy of classification estimated using the predictive distribution, and the estimation error of parameter mean values.
For the true mean value we take the estimate obtained by running Gibbs algorithm for $5000\cdot n$ steps, where $n$ is the number of parameters.

In Fig. \ref{fig:logreg}, we demonstrate the performance of all algorithms measured by accuracy and estimation error.
In terms of iterations the dynamical Gibbs algorithm performs better than others.
However, if we take the computational efforts into account the difference in performance will be insignificant (see Appendix \ref{app:experiments_logreg}).
It can be explained by the high computational cost of dynamical Gibbs: at every iteration we need to evaluate the conditionals for every dimension, while Gibbs requires only one conditional per iteration.
Note that this is not a problem for the Ising model since there we can efficiently evaluate all of the conditionals at once using matrix operations.
Thus, we conclude that our algorithm is favourable in the cases where the efficient evaluation of conditional distributions is possible or if we can compute them in parallel.

\section{Conclusion}
\label{sec:conclusion}

We propose a novel approach to sampling, relying on a purely deterministic dynamics defined by an ordinary differential equation.
To demonstrate the utility of our approach we derive the dynamical Gibbs algorithm that performs well on several tasks and is even able to outperform the independent sampler.
Finally, we characterize the family of such terministic samplers via the family of divergence-less vector fields on $\mathbb{R}^n$ and compact manifolds.

There are many interesting future directions to explore.  What can we prove about ergodicity and how do we design divergence free vector fields that mix fast? Under what conditions do we achieve $1/T$ convergence? Can we design hybrid samplers that, like HMC, alternate deterministic trajectories with stochastic moves? More generally, we believe that this work forges an interesting connection between machine learning and chaos theory that deserves further exploration.

\bibliography{bibliography}
\bibliographystyle{icml2021}


\newpage
\appendix
\onecolumn
\section{Proofs}
\subsection{Liouville's physical proof}
\label{app:liouville_proof}

To make the relation with sampling more apparent, we derive this equation from first principles.
That is, measure-preserving property of a function $g(x_0,t)$ could be obtained via the change of variables formula and written as
\begin{align}
    p(x_0) = p(g(x_0,t)) \bigg|\deriv{g(x_0,t)}{x_0}\bigg|,
    \;\text{ where }\; g(x_0,t) = x_0 + \int_{0}^t v(x(t)) dt.
    \label{eq_app:measure-preserving}
\end{align}
Here $g(x_0,t)$ is the evolution flow defined by the equation $dx/dt = v(x)$ that moves point $x_0$ for time $t$.
Note that we want our flow to preserve the measure for any time $t$, but we can consider only small times since if $p(x,t+dt) = p(x,t)$ then the density is stationary and doesn't change in time at all.
For small time $t$ we can write the time evolution as
\begin{align}
    g(x_0,t) = g(x_0,0) + \deriv{g(x_0,t)}{t}\bigg|_{t=0}t + o(t) = x_0 + v(x_0)t + o(t).
\end{align}
Then its Jacobian could be evaluated as
\begin{align}
    \bigg|\deriv{g(x_0,t)}{x_0}\bigg| = 1 + t\text{div}v(x_0) + o(t).
\end{align}
Putting this into \eqref{eq_app:measure-preserving}, we get
\begin{align}
    &p(x_0) = p(g(x_0,t)) + tp(g(x_0,t))\text{div}v(x_0) + o(t), \\ &\frac{dp}{dt} + p(x_0)\text{div}v(x_0) = 0 \implies \bigg(\nabla p,  \frac{dx}{dt}\bigg) + p(\nabla, v) = \text{div}(pv) = 0.
\end{align}

\subsection{Scaling lemma}
\label{app:scaling_lemma}

Here we prove the following lemma.
\begin{lemma*}
Define the phase curve of a trajectory $x(t)$ as the curve that the trajectory draws on the manifold on which $x(t)$ lives.
For any continuous vector field $v(x) \in \mathbb{R}^n$ and positive scalar function $s(x)$, the solutions of the following Cauchy problems
\begin{align}
    \frac{\dd x}{\dd t} = v(x), \;\mathrm{ and }\; \frac{\dd x}{\dd t} = s(x)v(x), \;\; s(x) > 0 \;\; \forall x, \;\; x(t_0) = x_0
\end{align}
have the same phase curves.
\end{lemma*}
\begin{proof}
Denote the solutions of both equations as $\phi(t)$ for the equation $\dot{x} = v$ and $\psi(t)$ for the equation $\dot{x} = sv$.
Then for both solutions we have
\begin{align}
    \frac{\dd \phi}{\dd t} = v(\phi(t)), \;\;\; \frac{\dd \psi}{\dd t} = s(\psi(t))v(\psi(t)).
\end{align}
Assume there exist such scalar function $\kappa(t)$ that $\psi(\kappa(t))$ is the solution of $\dot{x} = v$, and $\kappa(t_0) = t_0$.
If this is the case, then this function must satisfy
\begin{align}
    \frac{\dd \psi}{\dd t}= v(\psi(\kappa(t))) \;\; \implies \;\; \frac{\partial \psi}{\partial \kappa} \bigg|_{\kappa(t)} \frac{\dd \kappa}{\dd t} = s(\psi(\kappa(t)))v(\psi(\kappa(t))) \frac{\dd \kappa}{\dd t} = v(\psi(\kappa(t))) \;\; \implies \;\; \frac{\dd \kappa}{\dd t} = \frac{1}{s(\psi(\kappa(t)))}.
\end{align}
Clearly, such function $\kappa(t)$ exists and unique since it is the solution of the following Cauchy problem
\begin{align}
    \frac{\dd \kappa}{\dd t} = \frac{1}{s(\psi(\kappa))}, \;\;\; \kappa(t_0) = t_0.
\end{align}
Moreover, this function is monotonous since $s(x) > 0 \;\; \forall x$.
Thus, we have demonstrated that both functions $\phi(t)$ and $\psi(\kappa(t))$ are the solutions of the same Cauchy problem ($\dot{x} = v, \; x(t_0) = x_0$), where $\kappa$ is just a monotonous rescaling of time, i.e. a bijective change of the time variable.
\end{proof}

\subsection{Discontinuous Liouville's theorem}
\label{app:disc_liouville}

\begin{figure}[h]
    \centering
    \begin{subfigure}[b]{0.45\textwidth}
            \includegraphics[width=\linewidth]{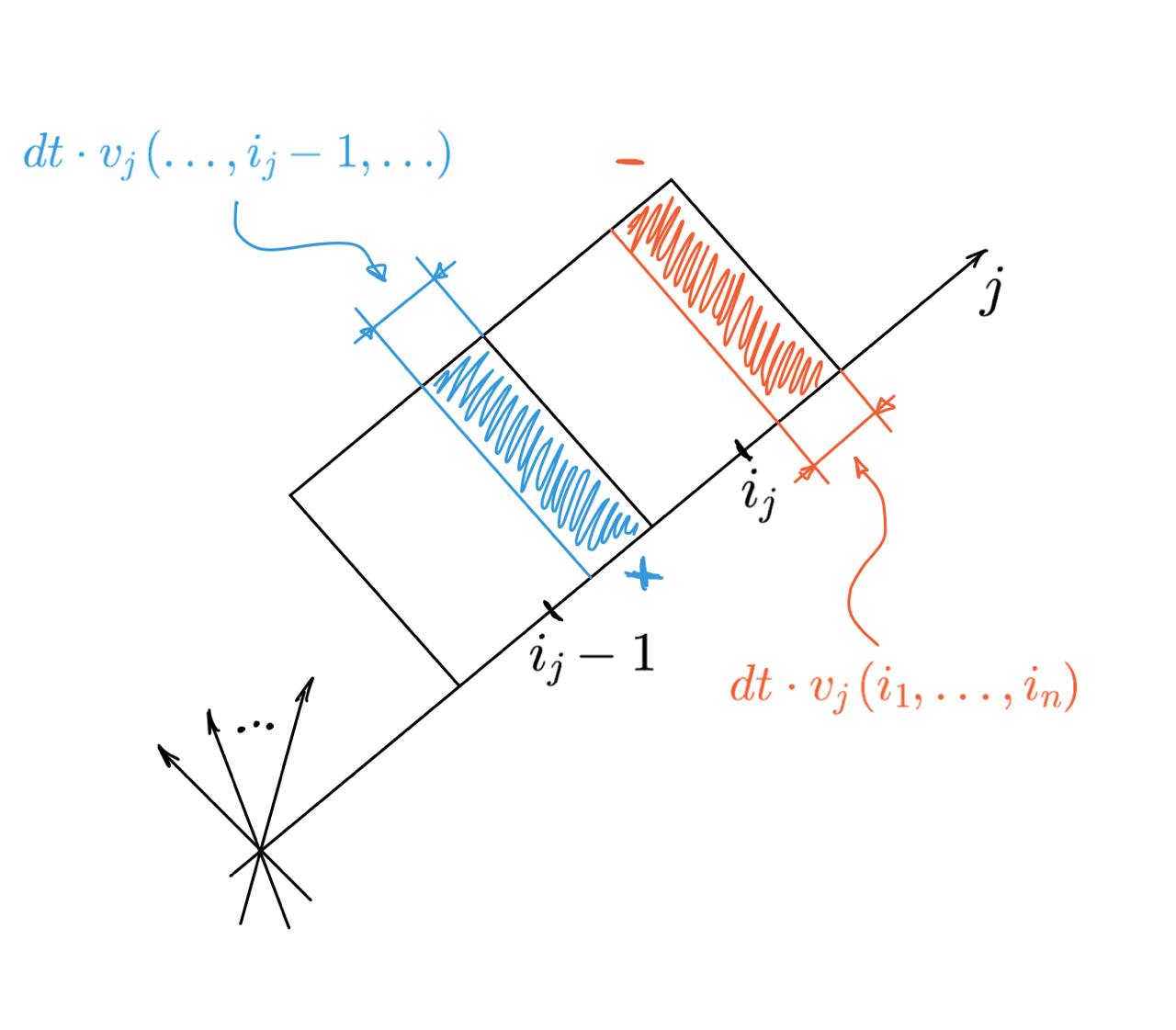}
            \caption{Change of volume along $j$-th dimension.}
            \label{fig:disc_liouville_1}
    \end{subfigure}%
    \begin{subfigure}[b]{0.45\textwidth}
            \includegraphics[width=\linewidth]{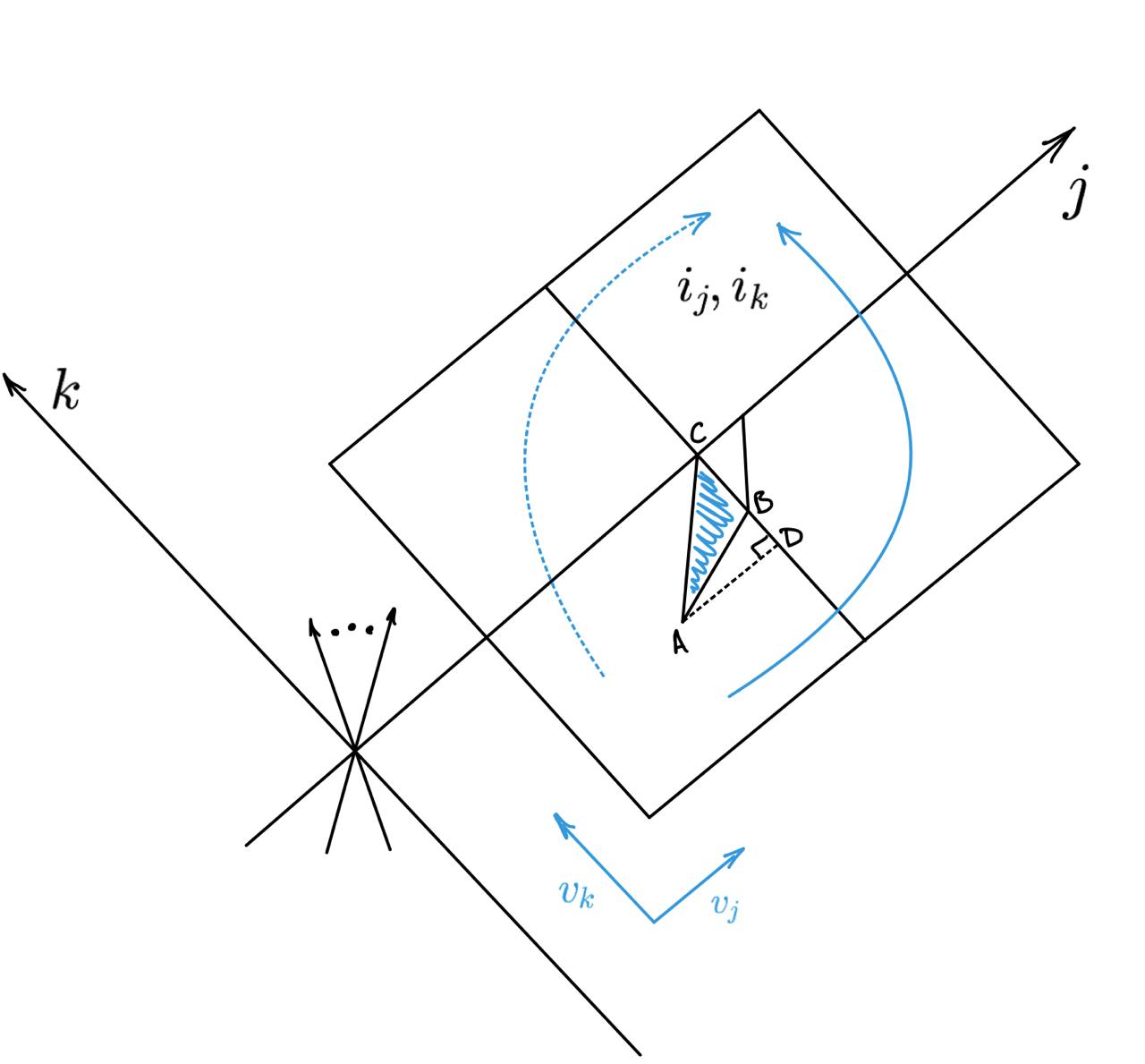}
            \caption{Flow of volume between dimensions $k$ and $j$.}
            \label{fig:disc_liouville_2}
    \end{subfigure}
    \caption{Illustrations for the discontinuous Liouville's theorem.}
    \label{fig:disc_liouville}
\end{figure}

Here we consider the proof of the Liouville's theorem for the dynamical Gibbs operating on a discrete distribution.
As we describe in the text we consider the discrete distribution as a joint distribution of $n$ discrete random variables and the $j$-th random variable takes integer values in $[0,\ldots,n_j-1]$.
The target distribution is then dequantized in the following sense.
Every state with indices $i_1,\ldots,i_n$ is associated with the hyper cube $\{(i_1,\ldots,i_n)+u: u\in[0,1]^n\}$.
Inside this hypercube the density is constant and equals to $p(i_1,\ldots,i_n)$.
The vector field $v$ is also constant in each cell and is denoted as $v(i_1,\ldots,i_n)$.

For the given vector field we consider the change of the density in time:
\begin{align}
\begin{split}
    p(i_1,\ldots,i_n,t+dt) - p(i_1,\ldots,i_n,t) = & \sum_{j=1}^ndtv_j(\ldots,i_j-1,\ldots)p(\ldots,i_j-1,\ldots,t) - \\ 
    &- \sum_{j=1}^ndtv_j(i_1,\ldots,i_n)p(i_1,\ldots,i_n,t) + o(dt).
\end{split}
\label{appeq:liouville}
\end{align}
The first two summands correspond to the ingoing and outgoing mass, and are the first order terms w.r.t. $dt$. 
The illustration for these terms you can see at Fig. \ref{fig:disc_liouville_1}.
The flows from other directions are $o(dt)$ as illustrated by the example in Fig. \ref{fig:disc_liouville_2}.
Indeed the square of the triangle ABC is $S = \frac{1}{2}dtv_j(\ldots,i_k-1,\ldots,i_j-1,\ldots)dtv_k(\ldots,i_k-1,\ldots)$. 
The same reasoning applies for all combinations of dimensions.

From equation \eqref{appeq:liouville}, we get
\begin{align}
    \frac{d p(i_1,\ldots,i_n)}{dt} = - \bigg[\sum_{j=1}^ndtv_j(i_1,\ldots,i_n)p(i_1,\ldots,i_n,t) - \sum_{j=1}^ndtv_j(\ldots,i_j-1,\ldots)p(\ldots,i_j-1,\ldots,t) \bigg],
\end{align}
which is the discrete analog of the continuity equation.

For the dynamical Gibbs we have,
\begin{align}
    v_j(i_1,\ldots,i_n) = c_j \frac{\sum_{i_j} p(i_1,\ldots,i_n)}{p(i_1,\ldots,i_n)} \;\;\implies\;\; \frac{d p(i_1,\ldots,i_n)}{dt} = 0.
\end{align}

\section{Hodge decomposition}
\label{app:hodge}

\subsection{Dynamical Gibbs}

We first remind that we can put the distribution on a torus $S^n$ continuously by mapping the distribution to the unit cube $[0,1]^n$ via hyperbolic tangent and then gluing the corresponding edges together.
Note that the density $p(x) = 0$ for all $x$ from the edges.
The vector field defining the measure-preserving flow on a manifold can be written as
\begin{align}
    v = \frac{1}{p} \omega^\sharp, \;\; \text{ where } \;\; \omega \in \Omega^1(S^n) \;\; \text{ and }\;\; \omega  = \gamma + \delta\beta, \;\; \gamma \in \text{ker}_\Delta(\Omega^1(S^n)), \;\; \beta \in \Omega^2(S^n).
\end{align}
The 1-form corresponding to the dynamical Gibbs is then $\omega = \sum_i p(x_{\setminus i}) dx_i$. 
Here we provide $\beta$ such that $\omega = \delta \beta$.
Indeed, straightforward evaluation shows that $d\omega \neq 0$, hence, $\beta \neq 0$.
The equality $\omega = \delta \beta = \star d \star \beta$ can be rewritten as
\begin{align}
    \star \omega = \sum_{i=1}^n (-1)^{i+1} p(x_{\setminus i}) \bigwedge_{j\neq i} dx_j = d\star\beta, \;\; \text{ where } \;\; (\star\beta) \in \Omega^{n-2}(S^n).
\end{align}
Now we see that the solution is
\begin{align}
    \star \beta = \sum_{i < j}\int^{x_j}_{0} dx_j' p(x_{\setminus i}) (-1)^{j+1}(-1)^{i+1} \frac{1}{n-i} \bigwedge_{k\neq i,j} dx_k.
\end{align}
Indeed,
\begin{align}
    d\star \beta & = \sum_{i < j} p(x_{\setminus i}) (-1)^{j+1}(-1)^{i+1} \frac{1}{n-i} dx_j \bigwedge_{k\neq i,j} dx_k = \\
    & = \sum_{i =1}^{n-1} \sum_{j=i+1}^n p(x_{\setminus i}) (-1)^{j+1}(-1)^{i+1} \frac{1}{n-i} (-1)^{j+1} \bigwedge_{k\neq i} dx_k = \\
    & = \sum_i \frac{n-i}{n-i} p(x_{\setminus i}) (-1)^{i+1} \bigwedge_{k\neq i} dx_k = \star \omega
\end{align}
Hence, the dynamical Gibbs sampler can be written as
\begin{align}
    v = \frac{1}{p}(\delta\beta)^\sharp, \;\;\text{ where }\;\; \beta = \star\bigg(\sum_{i < j}\int^{x_j}_{0} dx_j' p(x_{\setminus i}) (-1)^{j+1}(-1)^{i+1} \frac{1}{n-i} \bigwedge_{k\neq i,j} dx_k\bigg).
\end{align}

\subsection{Hamiltonian Monte Carlo}

The same reasoning applies for the Hamiltonian Monte Carlo.
The corresponding 1-form is
\begin{align}
    \omega = \sum_{i=1}^n \frac{\partial H}{\partial p_i} dq_i - \sum_{i=1}^n \frac{\partial H}{\partial q_i} dp_i, 
\end{align}
hence,
\begin{align}
    \star\omega = \sum_{i=1}^n \frac{\partial H}{\partial p_i} (-1)^{i+1}\bigwedge_{j\neq i} dq_j \bigwedge_k dp_k - \sum_{i=1}^n \frac{\partial H}{\partial q_i} (-1)^{i+1} (-1)^n\bigwedge_{k} dq_k \bigwedge_{j\neq i} dp_j.
\end{align}
It is easy to check that this differential form corresponds to divergence-less vector field by verifying $d\star \omega = 0$.
As in the previous case with the dynamical Gibbs, we want to find such $\beta$ that $\star \omega = d\star \beta$.
Thus, we see
\begin{align}
    \star\beta = \sum_i (-1)^{n-1}H \bigwedge_{k\neq i}dq_k\bigwedge_{l\neq i}dp_l.
\end{align}
Indeed,
\begin{align}
    d\star \beta = \sum_i\bigg((-1)^{n-1}\frac{\partial H}{\partial q_i} (-1)^{i+1}\bigwedge_k dq_k \bigwedge_{l\neq i}dp_l + (-1)^{n-1}\frac{\partial H}{\partial p_i} (-1)^{i+1}(-1)^{n-1}\bigwedge_{k\neq i} dq_k \bigwedge_{l}dp_l\bigg) = \star\omega.
\end{align}
Hence, the Hamiltonian Monte Carlo can be written as
\begin{align}
    v = \frac{1}{p}(\delta\beta)^\sharp, \;\;\text{ where }\;\; \beta = \star\bigg(\sum_i (-1)^{n-1}H \bigwedge_{k\neq i}dq_k\bigwedge_{l\neq i}dp_l\bigg).
\end{align}


\section{Experiments}
\label{app:experiments}

\subsection{Gibbs portrait}
\label{app:experiments_portrait}

\begin{figure}[h]
    \centering
    \includegraphics[width=0.95\textwidth]{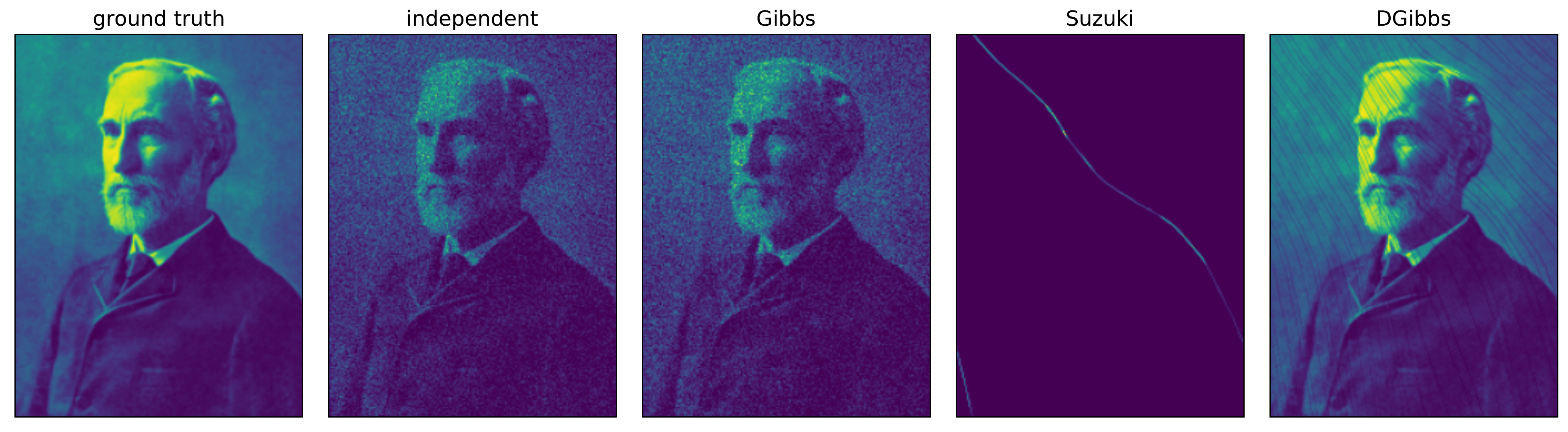}
    \caption{2D histograms for different samplers operating on the 2D target discrete distribution (first picture).}
    \label{fig:ising_app}
\end{figure}

\subsection{Ising model}
\label{app:experiments_ising}

\begin{figure}[h]
    \centering
    \includegraphics[width=0.24\textwidth]{pics/log_prob_ising.pdf}
    \includegraphics[width=0.24\textwidth]{pics/error_ising.pdf}
    \includegraphics[width=0.24\textwidth]{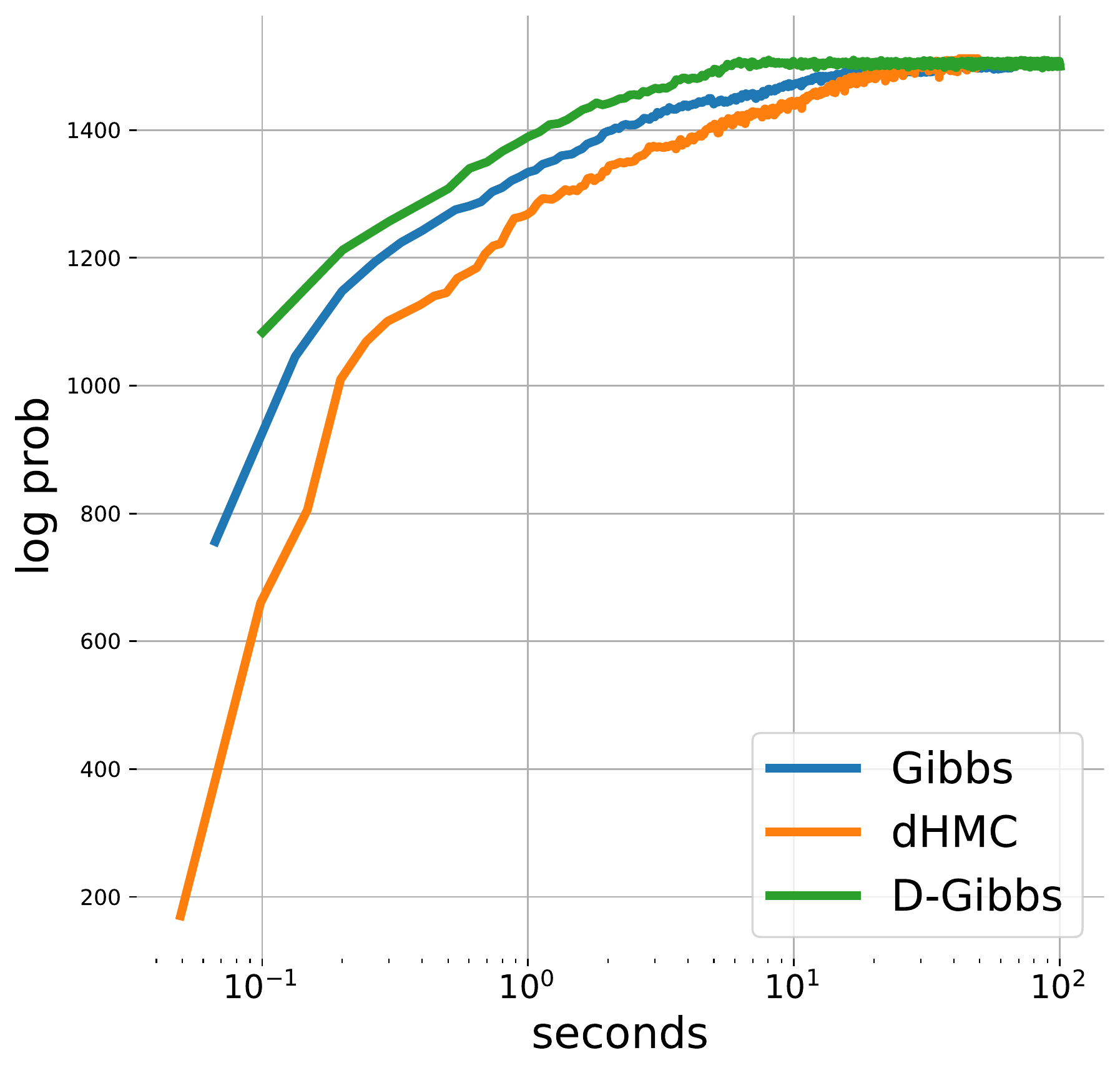}
    \includegraphics[width=0.24\textwidth]{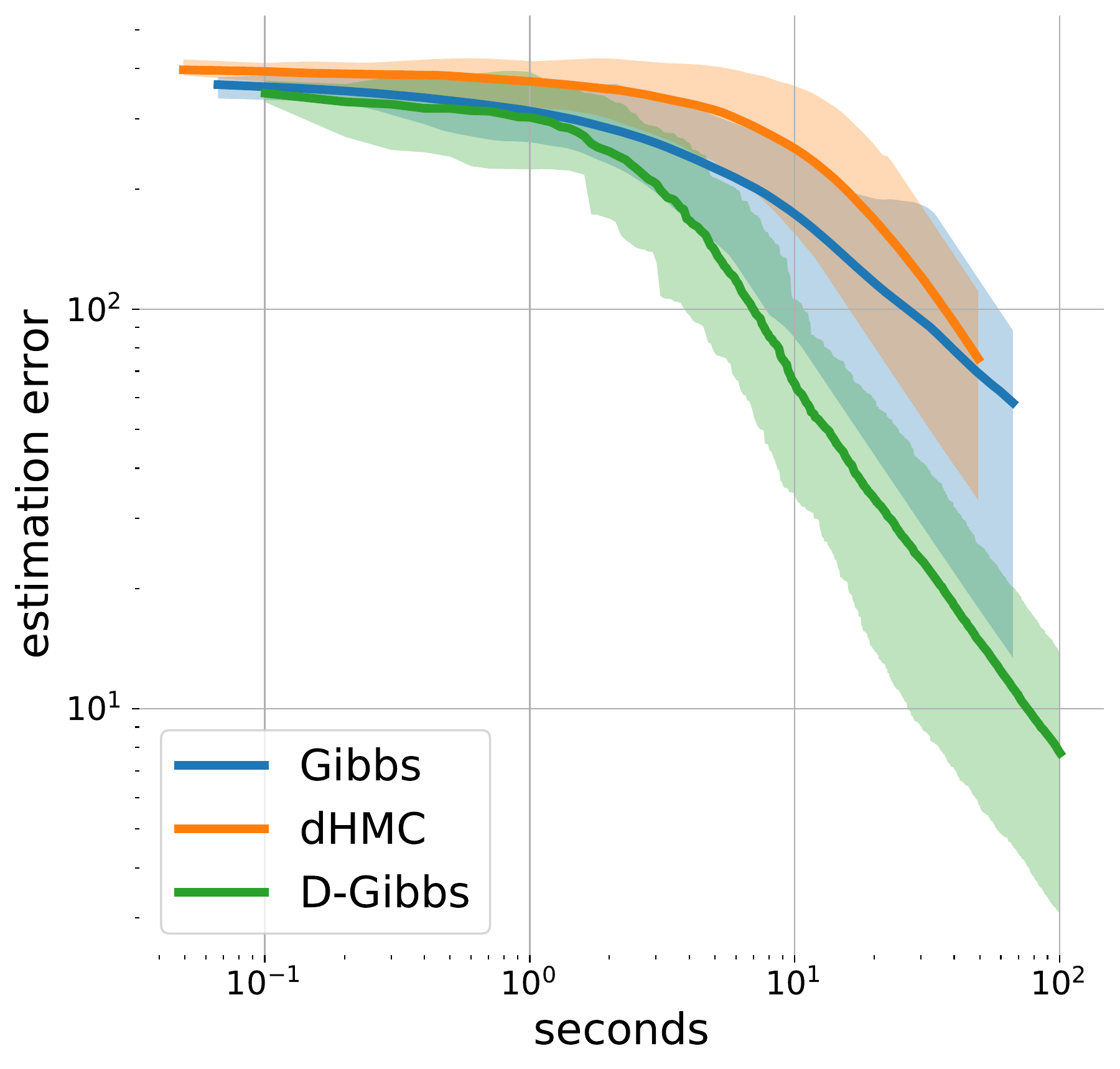}
    \caption{Convergence of samplers for the Ising model, in terms of iterations on the left and in terms of running time on the right. For the details of these experiments see code provided in the supplementary material.}
    \label{fig:ising_app}
\end{figure}

\subsection{Binarized logistic regression}
\label{app:experiments_logreg}

\begin{figure*}[h]
    \centering
    \includegraphics[width=0.24\textwidth]{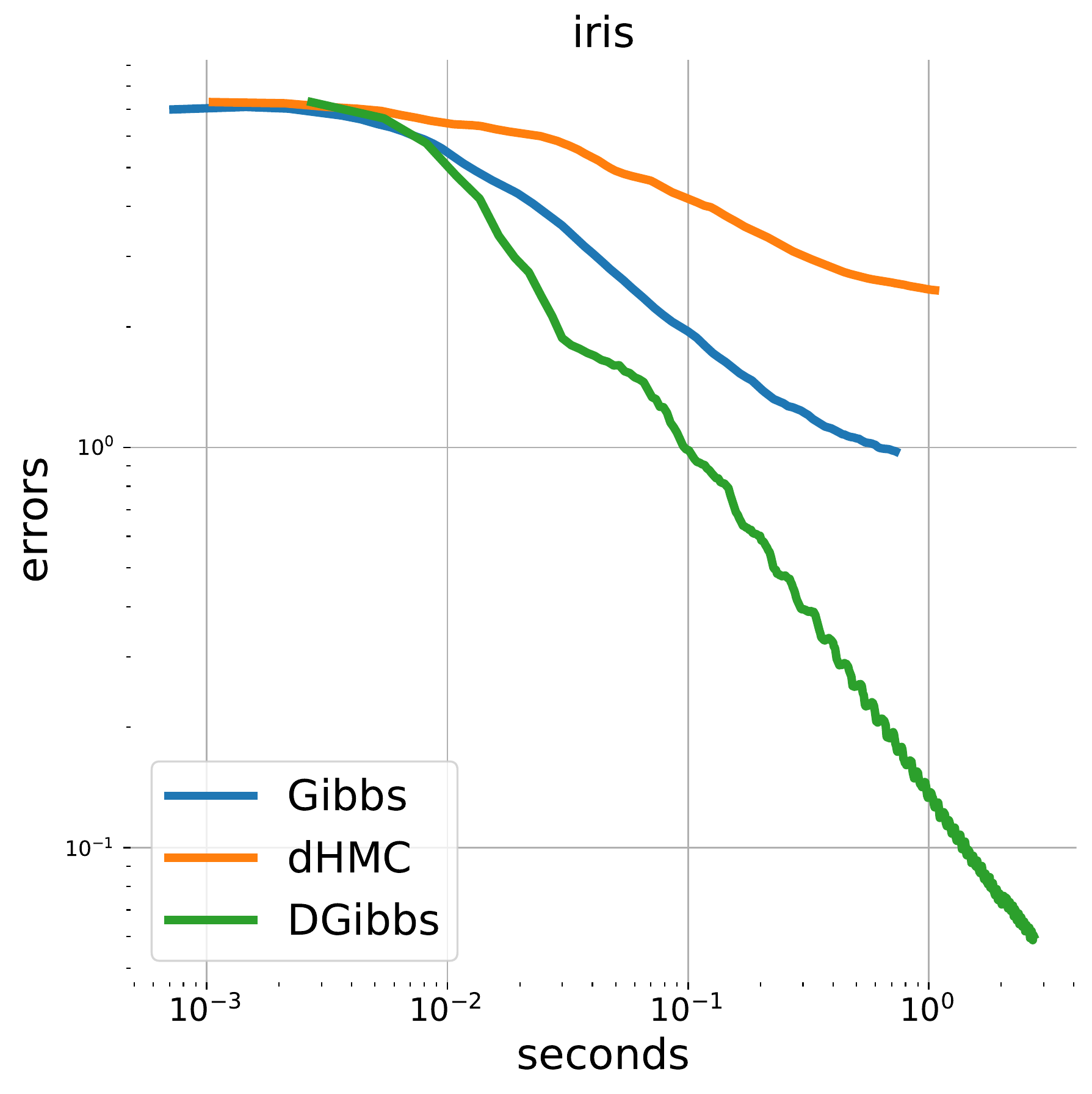}
    \includegraphics[width=0.24\textwidth]{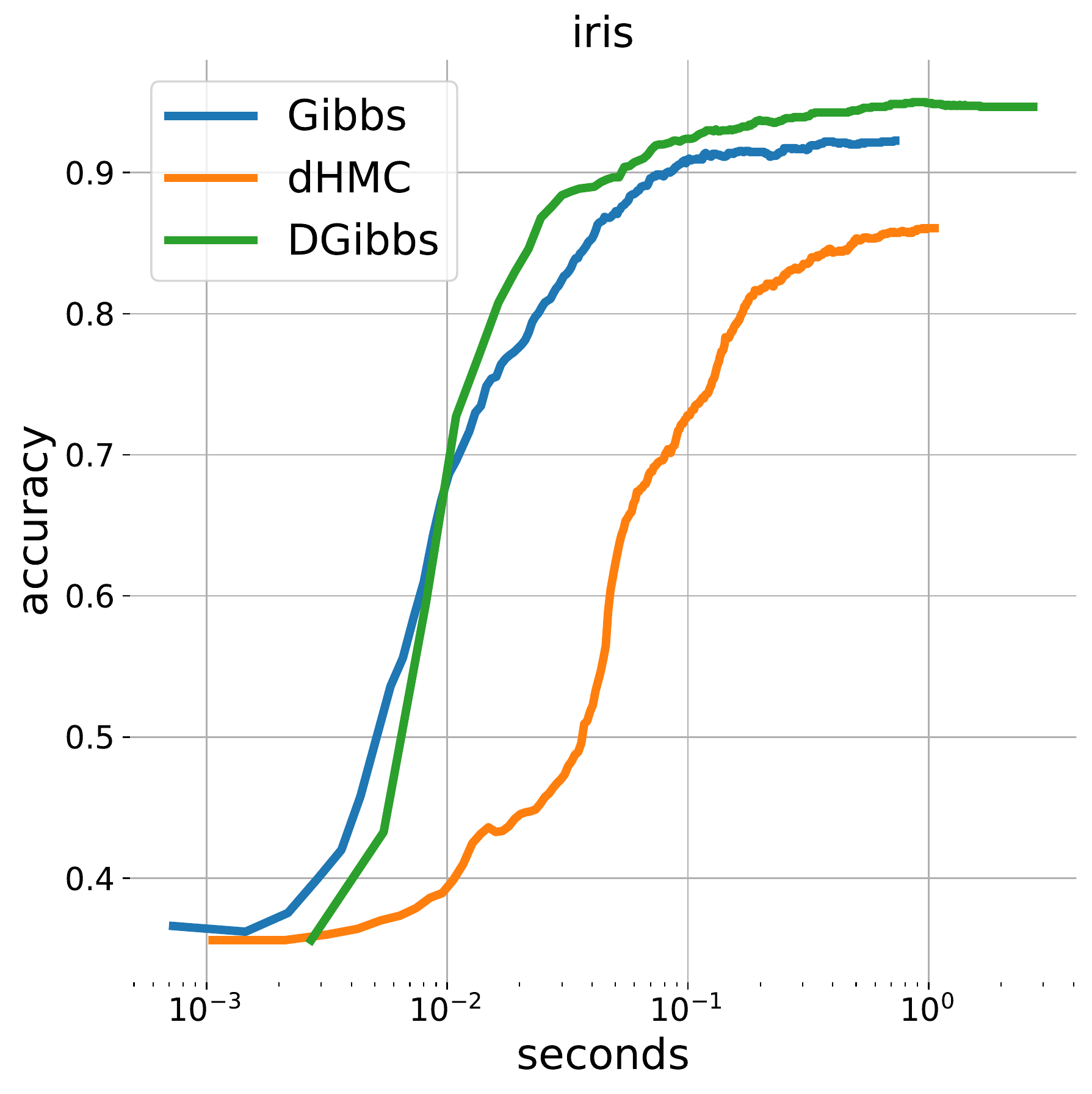}
    \includegraphics[width=0.24\textwidth]{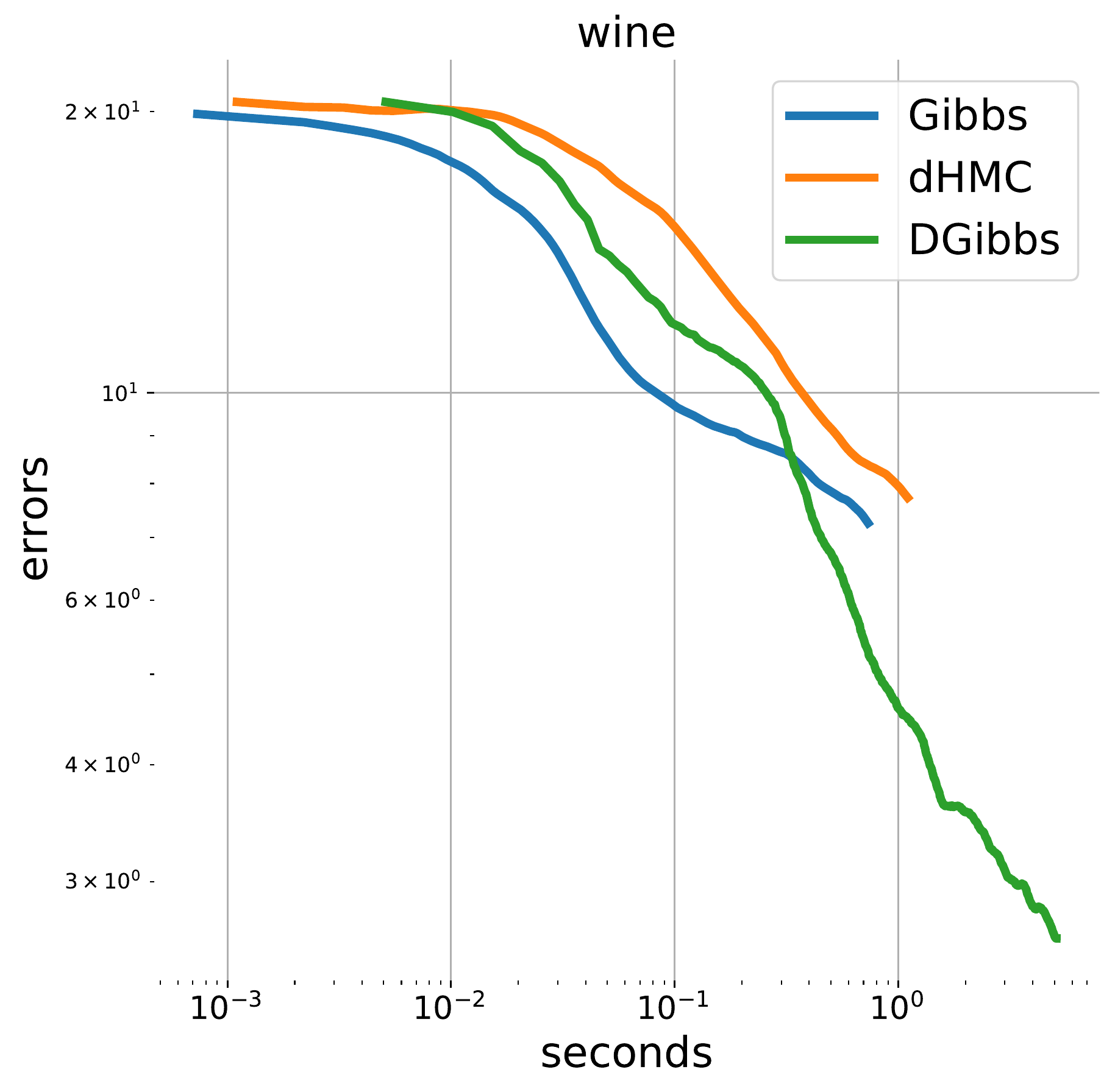}
    \includegraphics[width=0.24\textwidth]{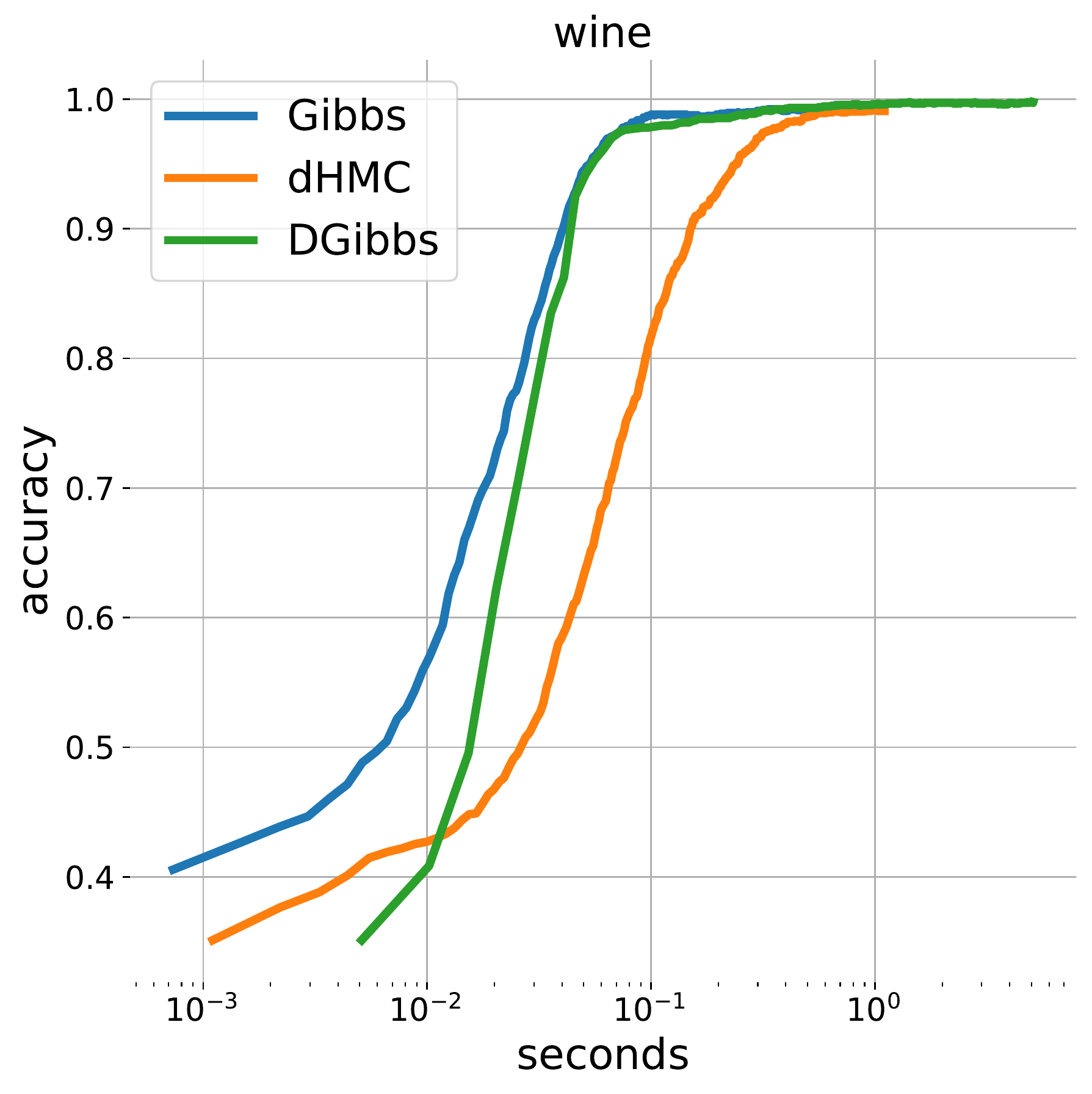}
    \caption{Sampling from the posterior distribution of logistic regression with binarized parameters. As datasets we consider Iris (on the left) and Wine (on the right). For both datasets we report the accuracy of classification estimated by averaging across collected samples and the error in the estimation of mean parameters values. For each metric we report running time on the horizontal axis.}
    \label{fig:logreg}
\end{figure*}

\end{document}